\documentclass[10pt]{article}
\makeatletter
\renewcommand{\thefootnote}{}

\usepackage{epsfig,amsmath,graphicx,amssymb,overpic}

\usepackage{graphicx}
\usepackage{dcolumn}
\usepackage{bm}
\usepackage{graphicx}
\usepackage{subfigure}
\usepackage{epsfig,amsmath,graphicx,amssymb,overpic,cite}

\usepackage{graphicx}
\usepackage{dcolumn}
\usepackage{bm}
\usepackage{graphicx}
\usepackage{subfigure}
\usepackage{amsthm}

\usepackage{pdfsync}
\usepackage{cite}
\usepackage{amscd}
\usepackage{amsmath}
\usepackage{latexsym}
\usepackage{amsfonts}
\usepackage{amssymb}
\usepackage{amsthm}
\usepackage{graphicx}
\usepackage{indentfirst}
\usepackage{enumerate}
\usepackage{amstext}
\usepackage[dvipdfm,
            pdfstartview=FitH,
            CJKbookmarks=true,
            bookmarksnumbered=true,
            bookmarksopen=true,
            colorlinks=true, 
            pdfborder=001,   
            citecolor=magenta,
            linkcolor=blue,
            ]{hyperref}

 \allowdisplaybreaks[4]

\newtheorem{lemma}{Lemma}[section]
\newtheorem{propo}{Proposition}[section]
\newtheorem{prop}{RHP}

\usepackage{tikz}
\usetikzlibrary{shapes,arrows,positioning}
\usepackage{xcolor} 

\def\be{\begin{equation}}
\def\ee{\end{equation}}
\def\bee{\begin{eqnarray}}
\def\ene{\end{eqnarray}}
\def\bes{\begin{subequations}}
\def\ees{\end{subequations}}
\def\det{{\rm det}}
\def\d{\displaystyle}
\def\v{\vspace{0.05in}}
\def\no{{\nonumber}}

\begin{document}

\baselineskip=13pt
\renewcommand {\thefootnote}{\dag}
\renewcommand {\thefootnote}{\ddag}
\renewcommand {\thefootnote}{ }

\pagestyle{plain}

\begin{center}
\baselineskip=16pt \leftline{} \vspace{-.3in} {\Large \bf Large-space and long-time asymptotic behaviors of $N_{\infty}$-soliton solutions for the focusing Hirota equation} \\[0.2in]
\end{center}

\begin{center}
{\bf Weifang Weng$^{a}$,\,\, Zhenya Yan}$^{b,c,*}$\footnote{$^{*}${\it Email address}: zyyan@mmrc.iss.ac.cn (Corresponding author)}  \\[0.08in]
{\small \it$^a$Department of Mathematical Sciences, Tsinghua University, Beijing 100084, China} \\
\small \it $^b$KLMM, Academy of Mathematics and Systems Science, Chinese Academy of Sciences, Beijing 100190, China\\
\small\it $^c$School of Mathematical Sciences, University of Chinese Academy of Sciences, Beijing 100049, China \\[0.18in]
\end{center}

\v\v
\noindent {\bf Abstract}\, {\small The Hirota equation is one of the higher-order extensions of the nonlinear Schr\"odinger equation, and can describe  the ultra-short optical pulse propagation in the form $
iq_t+\alpha(q_{xx}+ 2|q|^2q)+i\beta (q_{xxx}+ 6|q|^2q_x)=0,\, (x,t)\in\mathbb{R}^2\, (\alpha,\,\beta\in\mathbb{R})$. In this paper, we analytically explore the large-space and long-time asymptotic behaviors of a soliton gas (the limit $N\to \infty$ of $N$-soliton solutions) for the Hirota equation including the complex modified KdV equation by using the Riemann-Hilbert problems with discrete spectra restricted in  intervals $(ia, ib)\cup (-ib, -ia)\, (0<a<b)$. We find that this soliton gas tends slowly to the Jaocbian elliptic wave solution with an error $\mathcal{O}(|x|^{-1})$ (zero exponentially quickly ) as
$x\to -\infty$ ($x\to +\infty$).
  We also present the long-time asymptotics of the soliton gas under the different velocity conditions: $x/t>4\beta b^2,\, \xi_c<x/t<4\beta b^2,\, x/t<\xi_c$. Moreover, we analyze the property of the soliton gas of the Hirota equation for the case of the discrete spectra filling uniformly a quadrature domain.}

\vspace{0.15in} \noindent {\bf Mathematics Subject Classification:} Primary 35Q51 $\cdot$ 35Q15 $\cdot$ 37K15;  Secondary 35C20




\section{Introduction}

Since Zabusky and Kruskal~\cite{soliton} first coined the concept of ``soliton" for the study of nonlinear wave propagations of the Korteweg-de Veris (KdV) equation with the periodic initial data in 1965, and Gardner-Greene-Kruskal-Miura~\cite{GGKM} first presented the inverse scattering transform (IST) method (alias nonlinear Fourier transform~\cite{AKNS}) to analytically solve the $N$-soliton solution of the KdV equation starting from its linear spectral problem (alias the Lax pair~\cite{Lax}) in 1967, more and more one-soliton and multi-soliton wave phenomena have been found to appear in a variety of physical models including the integrable and nearly-integrable nonlinear wave equations arising from many fields~\cite{AC91,ACN,FT,ACS,Nov84,exp,book1,book2,book3,nls05}. The $N$-soliton solutions solving integrable nonlinear wave equations via the IST provide us with the very good frameworks of nonlinear superpositions
to further analyze their elastic interactions and other related properties. For example, in 1971, Zahkarov~\cite{zak71} first intrdcued the concept of a {\it soliton gas}, which was defined as the limit of $N\to \infty$ of the $N$-soliton solution of the KdV equation. After that, the soliton gas was extended to study the soliton gas, breather gas, dense soliton gas, and hydrodynamic of other nonlinear wave equations~\cite{E1,E2,E3,E4,E5,E6,E7}. More recently, Grava {\it et al}~\cite{Girotti-1,Girotti-2,Grava-3} presented a new idea to study the soliton gas, dense soliton gas, and soliton shielding of the KdV, modified KdV, and NLS equations, respectively, starting from the $N$-soliton solution characterized by a Riemann-Hilbert problem.

Moreover, the larger-order limits of multipole solitons and rogue waves were studied~\cite{miller,bil1,bil2,bil3} by using the robust inverse scattering transform~\cite{RIST} and Riemann-Hilbert problems~\cite{RH2,RH}. The Deift-Zhou's nonlinear steepest descent method~\cite{RH} and $\bar\partial$-steepest descent method~\cite{dbar,dbar2} with Riemann-Hilbert problems~\cite{RHP1,RHP2,RHP3} were used to study the long-time asymptotics of some integrable nonlinear wave systems, for example, the nonlinear Schr\"odinger (NLS), KdV, Camassa-Holm (CH), Degasperis-Procesi, Fokas-Lenells, Sasa-Satsuma, modified CH, derivative NLS, and short-pulse equations (see ~\cite{n1a,n1,n2,n3,n4,n5,n6,n7,d1,d2,d3,d4,d5,d6,d7} and references therein).

To study more complicated nonlinear wave phenomena (e.g., the phenomena of oceanic  and optical waves), the extensions of the NLS equation were needed. For example, when the optical pulses become shorter (e.g., 100 fs~\cite{exp,book2}), higher-order dispersive and nonlinear effects such as third-order dispersion, self-frequency shift, and self-steepening arising from the stimulated Raman scattering are significant in the study of ultra-short optical pulse propagation to derive the multi-parameter higher-order NLS equation~\cite{hnls,hnls2,yan13,SS}
\bee\label{hirota-0}
iq_t+\alpha_1q_{xx}+ \alpha_2|q|^2q+i\epsilon(\beta_1 q_{xxx}+ \beta_2|q|^2q_x +\beta_3 q(|q|^2)_x)=0, \quad (x,t)\in\mathbb{R}^2,
\ene
where $q=q(x,t)$ is complex envelop field, $\alpha_{1,2}, \epsilon, \beta_{1,2,3}$ are real parameters. Generally speaking, Eq.~(\ref{hirota-0}) may not be completely integrable. However, some special cases of Eq.~(\ref{hirota-0}) can be solved via the IST if these parameters satisfy some constraint conditions. For example, when $\beta_1:\beta_2:\beta_3=0:0:0$, Eq.~(\ref{hirota-0}) reduces to the integrable focusing ($\alpha_1\alpha_2>0$) or defocusing ($\alpha_1\alpha_2<0$) NLS equation~\cite{ZS}
 \bee\label{hirota-01}
iq_t+\alpha_1q_{xx}+ \alpha_2|q|^2q=0.
\ene
 When $\beta_1:\beta_2:\beta_3=0:1:1$, Eq.~(\ref{hirota-0}) reduce to the the modified derivative NLS equation~\cite{wadati79}
  \bee\label{hirota-02}
iq_t+\alpha_1q_{xx}+ \alpha_2|q|^2q+i\epsilon\beta_2(|q|^2q)_x=0,
\ene
which, using a gauge transform, can be changed into the Kaup-Newell derivative NLS equation~\cite{KN}
 \bee\label{hirota-03}
iq_t+\alpha_1q_{xx}+i\epsilon(|q|^2q)_x=0.
\ene
 When $\beta_1:\beta_2:\beta_3=0:1:0$, Eq.~(\ref{hirota-0}) reduce to the the modified Chen-Lee-Liu derivative NLS equation~\cite{CLL,Kundu}
   \bee\label{hirota-04}
iq_t+\alpha_1q_{xx}+ \alpha_2|q|^2q+i\epsilon\beta_2|q|^2q_x=0,
\ene
When $\alpha_1:\alpha_2=1:2,\,\beta_1:\beta_2:\beta_3=1:6:3$, Eq.~(\ref{hirota-0}) reduce to the integrable Sasa-Satsuma equation~\cite{SS}
\bee\label{hirota-05}
iq_t+\alpha_1(q_{xx}+ 2|q|^2q)+i\epsilon \beta_1 (q_{xxx}+ 6|q|^2q_x +3 q(|q|^2)_x)=0,
\ene
When $\alpha_1:\alpha_2=1:2,\, \beta_1:\beta_2:\beta_3=1:6:0$, Eq.~(\ref{hirota-0}) reduce to the Hirota equation~\cite{hirota}
\bee\label{hirota}
iq_t+\alpha(q_{xx}+ 2|q|^2q)+i\beta (q_{xxx}+ 6|q|^2q_x)=0,
\ene
where $q=q(x,t)$ denotes the complex envelope field, the real-valued parameters $\alpha,\, \beta$ are used to modulate the second- and third-order dispersion coefficients, respectively. As $\alpha\not=0,\, \beta=0$, Eq.~(\ref{hirota}) reduces to the NLS equation, and at $\alpha=0,\, \beta\not=0$, Eq.~(\ref{hirota}) reduces to the complex modified KdV (cmKdV) equation
\bee\label{cmkdv0}
q_t+\beta (q_{xxx}+ 6|q|^2q_x)=0.
\ene
Notice that one can take $\beta=1$ by using the scaling transform $t=\tau/\beta$.

The Hirota equation (\ref{hirota}) and cmKdV equation (\ref{cmkdv0}) are the important ones of the above-mentioned higher-order NLS equations, and completely integrable, which can be derived from the AKNS hierarchy~\cite{AKNS}. In fact, they can also be found from the reduction of the Maxwell equation (see \cite{hnls2})
 \bee
\frac{\epsilon}{c^2}\frac{\partial^2 {\bf E}}{\partial t^2}+\nabla\times \nabla\times {\bf E}=0,
  \ene
where ${\bf E}$ denotes the electric field intensity, $c$ the speed of light, and $\epsilon$ a real-valued parameter.

In this paper, based on the idea of Refs.~\cite{Girotti-1,Girotti-2,Grava-3}, we would like to study the large-space and long-time asymptotic behaviors of a soliton gas of the integrable Hirota equation (\ref{hirota}) via the Riemann-Hilbert problems (see the following Riemann-Hilbert problem 2 in Sec. 2). The one-soliton solution of the Hirota equation (\ref{hirota}) is
\bee\label{ss}
 q(x,t)=k\, {\rm sech}[k(x-vt-x_0)]e^{i(\mu x+\omega t+p_0)}
\ene
with
\bee\no
 v=\beta(k^2-3\mu^2)+\frac{8}{3}\alpha\mu+\frac{2\alpha^2 \mu}{9\beta\mu-3\alpha},\quad \omega=k^2(\alpha-3\beta \mu)+\mu^2(\beta\mu-\alpha),
 \ene
where $x_0$ is the initial peak position of the soliton, $p_0$ the initial phase,  $v$ the soliton velocity,  $-\omega/\mu$ phase velocity, and $3k^2\beta-3\beta\mu^2+2\alpha\mu$ group velocity,
and the Jacobian elliptic double-periodic wave solution   of the Hirota equation (\ref{hirota}) is
\bee\label{ps}
 q(x,t)=k\, {\rm dn}[k(x-v_mt-x_0); m]e^{i(\mu x+\omega_m t+p_0)},\quad m\in (0,1)
\ene
with
\bee\no
 v_m=\beta[k^2(2-m^2)-3\mu^2]+\frac{8}{3}\alpha\mu+\frac{2\alpha^2 \mu}{9\beta\mu-3\alpha},\quad \omega_m=k^2(2-m^2)(\alpha-3\beta \mu)+\mu^2(\beta\mu-\alpha).
 \ene
where the elliptic wave velocity is $v_m$,  phase velocity $-\omega_m/\mu$, and group velocity $3(2-m^2)k^2\beta-3\beta\mu^2+2\alpha\mu$.
When $m\to 1$, the periodic solution (\ref{ps}) reduces to the one-soliton solution (\ref{ss}).

The Hirota equation and its higher-order extensions have been showed to possess the multi-soliton solutions, breather solutions, and rogue waves~\cite{hirota,h1,h2,h3,h4,h5,zhang20}. Moreover, the long-time asymptotic behaviors of the focusing and defocusing Hirota equations have been studied~\cite{hs1,hs2,hs3}. Recently, Guo {\it et al}~\cite{guo23} studied nonlinear stability of multi-solitons for the Hirota equation.

It is well known that one can solve the $N$-soliton solution of the Hirota equation (\ref{hirota}) by using the inverse scattering transform to study the associated Lax pair~\cite{AKNS}
\begin{align}\label{lax}
\begin{aligned}
&\frac{\partial \Psi}{\partial x}\!=\!\!\left(\!\!\!\begin{array}{cc} -iz & q(x,t)  \vspace{0.05in}\\
-q^*(x,t) & iz\end{array}\!\!\!\right)\!\Psi,\v\\
&\frac{\partial\Psi}{\partial t}\!=\!\!\left(\!\!\!\begin{array}{cc} -i(4\beta z^3\!+\!2\alpha z^2)\!+\!i(2\beta z\!+\!\alpha)|q|^2\!+\!2i\beta{\rm Im}(qq_x^*) & 4\beta z^2 q+2z(\alpha q-i\beta q_x)+T_{01}  \vspace{0.05in}\\
-4\beta z^2 q^*-2z(\alpha q^*-i\beta q_x^*)+T_{02} & i(4\beta z^3\!+\!2\alpha z^2)\!-\!i(2\beta z\!+\!\alpha)|q|^2\!-\!2i\beta{\rm Im}(qq_x^*)\end{array}\!\!\!\right)\!\Psi,
\end{aligned}
\end{align}
where $\Psi=\Psi(x,t;z)$ is a $2\times 2$ matrix eigenfunction, $z\in\mathbb{C}$ a spectral parameter, the star denotes the complex conjugate, and
\bee \no
T_{01}=-i\alpha q-\beta (2|q|^2q+q_{xx}),\quad T_{02}=i\alpha q^*+\beta (2|q|^2q^*+q_{xx}^*).
\ene

Let the new matrix function $\Phi(x,t;z)$ be
\bee
\Phi(x,t;z)=\Psi(x,t;z)e^{i[zx+(4\beta z^3+2\alpha z^2)t]\sigma_3}.
\ene
Then based on the boundary-value condition of solitons $\lim_{|x|\to\infty}q(x,0)=0$, it follows from the Lax pair (\ref{lax}) that one know that~\cite{AKNS}
 \begin{equation}
     \Psi(x,t;z)\sim e^{-i[zx+(4\beta z^3+2\alpha z^2)t]\sigma_3}, \quad
     \Phi(x,t;z)\sim \mathbb{I}, \quad |x|\rightarrow \infty.
 \end{equation}
and the Lax pair (\ref{lax}) can be rewritten as
\begin{equation}\label{lax2}
 \begin{array}{l}
 \dfrac{\partial\Phi}{\partial x} + i k [\sigma_3, \Phi] = \!\left(\!\!\!\begin{array}{cc} 0 & q(x,t)  \vspace{0.05in}\\
-q^*(x,t) & 0\end{array}\!\!\!\right)\Phi,  \v\\
\dfrac{\partial \Phi}{\partial t}  + i(4\beta z^3+2\alpha z^2)[\sigma_3, \Phi]=
\left(\!\!\!\begin{array}{cc} i(2\beta z\!+\!\alpha)|q|^2\!+\!2i\beta{\rm Im}(qq_x^*) & 4\beta z^2 q+2z(\alpha q-i\beta q_x)+T_{01}  \vspace{0.05in}\\
-4\beta z^2 q^*-2z(\alpha q^*-i\beta q_x^*)+T_{02} & -i(2\beta z\!+\!\alpha)|q|^2\!-\!2i\beta{\rm Im}(qq_x^*)\end{array}\!\!\!\right) \Phi,
\end{array}
\end{equation}
which can be written as fully differential form such that the Jost solutions $\Phi_{+}(x,t;z)$ and $\Phi_{-}(x,t;z)$ can be written as follows:
\bee \label{jost}
\begin{array}{l}
       \Phi_{\pm}(x,t;z)  =\mathbb{I} -  \d\int_{x}^{\pm\infty} e^{iz(y-x)\sigma_3} \left(\!\!\!\begin{array}{cc} 0 & q(y,t)  \vspace{0.05in}\\
-q^*(y,t) & 0\end{array}\!\!\!\right) \Phi_{\pm}(y,t;z)e^{-iz(y-x)\sigma_3}dy,
\end{array}
\ene
with $\Phi_{\pm}(x,t;z)  =\mathbb{I}+\mathcal{O} \left(1/z\right),\, z\to \infty$, where $\sigma_3$ is the third Pauli matrix.

By using Abel lemma and the zero traces of coefficient matrices of the Lax pair (\ref{lax}), one knows that $\det \Psi_{\pm} (k,x,t)$ are independent of variable $x$ and $\det \Psi_{\pm}=1$.  Furthermore, $\Psi_{\pm}=\Phi_{\pm} e^{-i[zx+(4\beta z^3+2\alpha z^2)t]\sigma_3}$ are linearly dependent to lead to
\begin{equation}\label{scatter1}
     \Psi_{-}(x,t;z) = \Psi_{+}(x,t;z)S(k), \quad \det(S) =1,
\end{equation}
or
\begin{equation}\label{scatter2}
     \Phi_{-}(x,t;z) = \Phi_{+}(x,t;z) e^{-i[zx+(4\beta z^3+2\alpha z^2)t]\widehat{\sigma_3}}S(k), \quad \det(S) =1,
\end{equation}
where $S(z)=(s_{ij}(z))_{2\times 2}$ is a $2\times 2$ scattering matrix.
 \begin{align}\label{S-lie}
\begin{aligned}
s_{jj}(z)=(-1)^{j+1}\left|\Psi_{-j}(x, t; z), \Psi_{+(3-j)}(x, t; z)\right|
=(-1)^{j+1}\left|\Phi_{-j}(x, t; z), \Phi_{+(3-j)}(x, t; z)\right|,\quad j=1,2,\v\\
s_{(3-j)j}(z)=(-1)^{j}|\Psi_{-j}(x, t; z), \Psi_{+j}(x, t; z)|=(-1)^{j}|\Phi_{-j}(x, t; z), \Phi_{+j}(x, t; z)|,
\quad j=1,2,
\end{aligned}
\end{align}
where $\Psi_{\pm}=\left(\Psi_{\pm 1}(x,t;z), \Psi_{\pm 2}(x,t;z) \right)$, $\Phi_{\pm}=\left(\Phi_{\pm 1}(x,t;z), \Phi_{\pm 2}(x,t;z) \right)$, where  $\Psi_{\pm1}\, (\Phi_{\pm 1})$ and $\Psi_{\pm 2}\,(\Phi_{\pm2})$ represent two columns of $\Psi_{\pm}\,(\Phi_{\pm})$, respectively. There are the symmetry: $\Psi_{\pm}(z)=\sigma_2\Psi_{\pm}^*(z^*)\sigma_2$ and
 $S(z)=\sigma_2S^*(z^*)\sigma_2$, where $\sigma_2$ is the second Pauli matrix. It follows from Eq.~(\ref{jost}) that these functions $\Psi_{-1}, \Psi_{+2}, \Phi_{-1}, \Phi_{+2}, s_{11}(z)$ ($\Psi_{+1}, \Psi_{-2}, \Phi_{+1}, \Phi_{-2}, s_{22}(z)$ ) can be analytically extended to $\mathbb{C}_{+}$ ($\mathbb{C}_{-}$) and continuous to continuously to $\mathbb{C}_{+}\cup \mathbb{R}$ ($\mathbb{C}_{-}\cup \mathbb{R}$). However another two scattering coefficients $s_{12}(z),\, s_{21}(z)$ can not be analytically continued away from $\mathbb{R}$. Thus the reflection coefficient is defined as $\rho(z)=s_{21}(z)/s_{11}(z)$.

The discrete spectra of the Hirota equation are the set of all values $z\in\mathbb{C}\backslash\mathbb{R}$ such that they possess eigenfunctions in $L^2(\mathbb{R})$, and satisfy $s_{11}(z_n)=0$ for $Z_+=\{z_n|z_n\in\mathbb{C}_+\}$ and $s_{22}(z_n^*)=0$ for $Z_-=\{z_n^*|z_n^*\in\mathbb{C}_-\}$.
Let the sectionally meromorphic matrix $M^{(0)}(z)=M^{(0)}(x,t;z)$ be
\bee\label{M-0}
M^{(0)}(z)=\begin{cases}
\begin{array}{ll}\left(\!\!\begin{array}{cc}
\dfrac{\Phi_{-1}(x,t;z)}{s_{11}(z)}, & \Phi_{+2}(x,t;z)
\end{array}\!\!\right)
=\left(\!\!\begin{array}{cc}
\dfrac{\Psi_{-1}(x,t;z)}{s_{11}(z)}, & \Psi_{+2}(x,t;z)
\end{array}\!\!\right)e^{i\theta(x,t;z)\sigma_3}, & z\in \mathbb{C}_+,\vspace{0.05in}\\
\left(\!\!\begin{array}{cc}
\Phi_{+1}(x,t;z), & \dfrac{\Phi_{-2}(x,t;z)}{s_{22}(z)}
\end{array}\!\!\right)
=\left(\!\!\begin{array}{cc}
\Psi_{+1}(x,t;z), & \dfrac{\Psi_{-2}(x,t;z)}{s_{22}(z)}
\end{array}\!\!\right)e^{i\theta(x,t;z)\sigma_3}, & z\in \mathbb{C}_-.
\end{array}
\end{cases}
\ene
where the phase function is
\bee
\theta(x,t;z)=zx+(4\beta z^3+2\alpha z^2)t.
\ene
Then based on the scattering relation given by Eq.~(\ref{scatter1}) or (\ref{scatter2}), one can introduce the following Riemann-Hibert problem (RHP)~\cite{AKNS,d2,zhang20}:

\begin{prop}\label{RH-0} Find a $2\times 2$ matrix function $M^{(0)}(x,t;z)$ that satisfies:

\begin{itemize}

 \item {} Analyticity: $M^{(0)}(x,t;z)$ is meromorphic in $\{z|z\in\mathbb{C}\setminus\mathbb{R}\}$ and takes
continuous boundary values on $\mathbb{R}$;

 \item {} Jump condition: The boundary values on the jump contour $\mathbb{R}$ are defined as
 \bee
M^{(0)}_+(x,t;z)=M^{(0)}_-(x,t;z)V_0(x,t;z),\quad z\in\mathbb{R},
\ene
where the jump matrix is defined as
\bee \no
V_0(x,t;z)=\left[\!\!\begin{array}{cc}
1+|\rho(z)|^2& \rho^*(z)e^{-2i\theta(x,t;z)}  \vspace{0.05in}\\
\rho(z)e^{2i\theta(x,t;z)} & 1
\end{array}\!\!\right],
\ene
 \item {} Normalization: $M^{(0)}(x,t;z)\rightarrow\mathbb{I},\quad z\rightarrow\infty$.

\item {} Residue conditions:  $M^{(0)}(x,t;z)$ has simple poles at each discrete spectrum in $Z_+=\{z_n|z_n\in\mathbb{C}_+\}\cup Z_-=\{z_n^*|z_n^*\in\mathbb{C}_-\}$ with:
\begin{align}
\begin{aligned}
&\mathop\mathrm{Res}\limits_{z=z_n}M^{(0)}(x,t;z)=\lim\limits_{z\to z_n}M^{(0)}(x,t;z)\left[\!\!\begin{array}{cc}
0& 0  \vspace{0.05in}\\
c_ne^{2iz_nx+4i(\alpha z_n^2+2\beta z_n^3)t}& 0
\end{array}\!\!\right],\vspace{0.05in}\\
&\mathop\mathrm{Res}\limits_{z=z_n^*}M^{(0)}(x,t;z)=\lim\limits_{z\to z_n^*}M^{(0)}(x,t;z)\left[\!\!\begin{array}{cc}
0& -c_n^*e^{-2iz_n^*x-4i(\alpha z_n^{*2}+2\beta z_n^{*3})t}  \vspace{0.05in}\\
0& 0
\end{array}\!\!\right].
\end{aligned}
\end{align}
where the norming constants $c_n(z_n), c_n^*(z_n^*)$ satisfy
\bee
 c_n(z_n)=\frac{\hat c_n(z_n)}{s_{11}'(z_n)},\quad c_n^*(z_n^*)=\frac{\hat c_n^*(z^*_n)}{s_{22}'(z^*_n)}
\ene
with $\hat c_n(z_n),\,\hat c_n^*(z^*_n)$ be defined by
\bee
 \Psi_{-1}(x,t;z_n)=\hat c_n(z_n)\Psi_{+2}(x,t;z_n),\quad \Psi_{-2}(x,t;z_n)=-\hat c_n^*(z^*_n)\Psi_{+1}(x,t;z_n^*),
\ene
\end{itemize}
\end{prop}

Therefore, based on the  Plemelj's formulae, one can has
\begin{align}\label{RHP-jie}
\begin{array}{rl}
M^{(0)}(x, t; z)=&\!\!\d \mathbb{I}+\sum_{n=1}^{N}\left[\frac{\mathop\mathrm{Res}\limits_{z=z_n}M^{(0)}(x, t; z)}{z-z_n}+\frac{\mathop\mathrm{Res}\limits_{z=z_n^*}M^{(0)}(x, t; z)}{z-z^*_n}\right] \v\\
&\quad \d +\frac{1}{2\pi i}\int_{\mathbb{R}}\frac{M_-^{(0)}(x, t; \zeta)(\mathbb{I}-V_0(x,t;\zeta))}{\zeta-z}\,\mathrm{d}\zeta,\quad z\in\mathbb{C}\backslash\mathbb{R},
\end{array}
\end{align}
such that the solution $q(x,t)$ of the Hirota equation is given by using the solution $M^{(0)}(x,t;z)$ given by Eq.~(\ref{RHP-jie}) of the RHP-1
\bee\label{fanyan}
\begin{array}{rl}
q(x,t)\!\! & =2i\lim\limits_{z\rightarrow\infty}z(M^{(0)}(x,t;z))_{12} \v\\
  &=\d 2i\left(\mathop\mathrm{Res}\limits_{z=z_n^*}M^{(0)}(x, t; z)\right)_{12}
  -\d\frac{1}{\pi}\int_{\mathbb{R}}\left(M_-^{(0)}(x, t; \zeta)(\mathbb{I}-V_0(x,t;\zeta))\right)_{12}d\zeta.
\end{array}
\ene

In this paper, we will use the idea of Refs.~\cite{Girotti-1,Girotti-2,Grava-3} to analyze the asymptotic behaviors of a soliton gas for the Hirota equation. The rest of this paper is arranged as follows. In Sec. 2, based on the Riemann-Hilbert problem 1, we present the Riemann-Hilber problem related to the soliton gas of the Hirota equation, which is regarded as the limit $N\to \infty$ of its $N$-soliton solution, where the discrete spectra are chosen as the pure imaginary numbers within the intervals $(ia, ib)\cup(-ib,-ia)$. In Sec. 3, we present the large-space asymptotics of the soliton gas of the potential $q(x,0)$ of the Hirota equation as $x\to -\infty$. In Sec. 4, we give the corresponding large-space asymptotics of the soliton gas of the potential $q(x,t_0)$  of the Hirota equation for the fixed $t=t_0$ as $x\to -\infty$. In Sec. 5, we give the long-time asymptotics of the soliton gas of the potential $q(x,t)$  of the complex mKdV equation ($\alpha=0,\, \beta=1$) as $t\to +\infty$ at three fundamental spatial domains, $\xi=x/t>4\beta b^2,\, \xi_c<\xi=x/t<4\beta b^2,\, \xi=x/t<\xi_c$. In Sec. 6, we study the property of the soliton gas with the discrete spectra filling uniformly a quadrature domain. Finally, we give some conclusions and discussions in Sec. 7.

\section{Soliton gas: the limit $N\to\infty$ of $N$-soliton solution via RHP}

In this section, we first of all construct the pure-soliton Riemann-Hilbert problem corresponding to the $N$-soliton solution with reflectionless potential (i..e, $\rho(z)=0$) and pure imaginary discrete spectra set $Z_N:=\{iz_j,-iz_j\}_{j=1}^N$ with $z_j\in\mathbb{R}_+$ from the RHP \ref{RH-0}.

\begin{prop}\label{RH1} 
Find a $2\times 2$ matrix function $M(x,t;z)$ that satisfies the following properties:

\begin{itemize}

 \item {} Analyticity: $M(x,t;z)$ is analytic in $\mathbb{C}\backslash Z_N$, and has simple poles at $z\in Z_N$.

 \item {} Normalization: $M(x,t;z)=\mathbb{I}+\mathcal{O}(z^{-1}),\quad z\rightarrow\infty$;

\item {} Residue conditions: $M(x,t;z)$ has simple poles at each point in $Z_N$ with:
\begin{align}\label{Res1}
\begin{aligned}
&\mathrm{Res}_{z=iz_j}M(x,t;z)=\lim\limits_{z\to iz_j}M(x,t;z)\left[\!\!\begin{array}{cc}
0& 0  \vspace{0.05in}\\
c_je^{-2z_jx+(8\beta z_j^3-4i\alpha z_j^2)t}& 0
\end{array}\!\!\right],\vspace{0.05in}\\
&\mathrm{Res}_{z=-iz_j}M(x,t;z)=\lim\limits_{z\to -iz_j}M(x,t;z)\left[\!\!\begin{array}{cc}
0& -c_j^*e^{-2z_jx+(8\beta z_j^3+4i\alpha z_j^2)t}  \vspace{0.05in}\\
0& 0
\end{array}\!\!\right].
\end{aligned}
\end{align}

\end{itemize}
\end{prop}

The matrix function solution of the RHP ~\ref{RH1} can be written as
\bee \label{M}
 M(x,t;z)=\mathbb{I}+\sum_{j=1}^N\left(\frac{1}{z-z_j}\left[\!\!\begin{array}{cc}
u_j(x,t) & 0  \vspace{0.05in}\\ v_j(x,t) & 0 \end{array}\!\!\right]
+\frac{1}{z-z_j^*}\left[\!\!\begin{array}{cc}
0& -v_j^*(x,t)  \vspace{0.05in}\\ 0 & u_j(x,t) \end{array}\!\!\right]\right),
\ene
where $u_j(x,t),\, v_j(x,t)$ are unknown functions to be determined later. Substituting Eq.~(\ref{M}) into the first one of Eq.~(\ref{Res1}) yields the system of equations about unknowns
$u_j, \, v_j^*$
\bee\label{uv}
\begin{array}{l}
 u_j=\d -\sum_{s=1}^N\frac{c_se^{-2z_sx+(8\beta z_s^3-4i\alpha z_s^2)t}}{z_j-z_s^*}v_s^*, \quad j=1,2,...,N, \v\\
 v_j^*=\d \left(1+\sum_{s=1}^N\frac{c_se^{-2z_sx+(8\beta z_s^3-4i\alpha z_s^2)t}}{z_j^*-z_s}u_j\right)c_j^*e^{-2z_sx+(8\beta z_s^3+4i\alpha z_s^2)t},\quad j=1,2,...,N.
 \end{array}
 \ene

Let $\widehat u=(\widehat u_1, \widehat u_2,..., \widehat u_N)^T,\, \widehat v^*=(\widehat v_1^*, \widehat v_2^*,...,\widehat v_N^*)^T$ and $w=(w_1,w_2,...,w_N)^T,\, Q=(Q_{js})_{N\times N}$ with
\bee \label{BT}
\begin{array}{c}
 \widehat u_j=\dfrac{u_j}{\sqrt{c_j}} e^{z_jx-(4\beta z_j^3-2i\alpha z_j^2)t},\quad
 \widehat v_j^*=\dfrac{v_j^*}{\sqrt{c_j^*}} e^{z_jx-(4\beta z_j^3+2i\alpha z_j^2)t}, \\
 w_j={\sqrt{c_j^*}} e^{-z_jx+(4\beta z_j^3+2i\alpha z_j^2)t}, \quad
 Q_{js}=\dfrac{i}{z_s-z_j^*}e^{-(z_j+z_s)x+[4\beta (z_j^3+z_s^3)+2i\alpha (z_j^2-z_s^2)]t}.
\end{array}
\ene
Then system (\ref{uv}) can be rewritten as
\bee
\left[\begin{array}{cc}
\mathbb{I}_N & -iQ^* \\
-iQ & \mathbb{I}_N
\end{array}\right]\left[\begin{array}{cc}
\widehat u \\  \widehat v^* \end{array}\right]
=\left[\begin{array}{cc}
 0 \\ w \end{array}\right],
\ene
whose solution can be derived in the form (Note that ${\rm det}(\mathbb{I}_N+QQ^*)\not=0$)
\bee\label{sys-s}
\left[\begin{array}{cc}
\widehat u \v\\  \widehat v^* \end{array}\right]=
\left[\begin{array}{cc}
\mathbb{I}_N-Q^*(\mathbb{I}_N+QQ^*)^{-1}Q & iQ^*(\mathbb{I}_N+QQ^*)^{-1} \v \\
i(\mathbb{I}_N+QQ^*)^{-1}Q & (\mathbb{I}_N+QQ^*)^{-1}
\end{array}\right]\left[\begin{array}{cc}
 0 \v\\ w \end{array}\right].
\ene

According to Eqs.~(\ref{fanyan}), (\ref{M}), (\ref{BT}), and (\ref{sys-s}), we can recover $q(x,t)$ by the following formula of the $N$-soliton solution of the Hirota equation (\ref{hirota}):
\bee\label{fanyan-1}
\begin{array}{cl}
q_N(x,t)&= 2i\lim\limits_{z\rightarrow\infty}(z M(x,t;z))_{12}, \v\\
 &= \d -2i\sum_{j=1}^N v_j^* \v\\
 &=\d -2i\sum_{j,s=1}^N((\mathbb{I}_N+QQ^*)^{-1})_{js}\sqrt{c_j^*c_s^*}e^{-(z_j+z_s)x+[4\beta(z_j^3+z_s^3)+2i\alpha(z_j^2+z_s^2)]t}.
 \end{array}
\ene

\begin{figure}[!t]
    \centering
 \vspace{-0.15in}
  {\scalebox{0.32}[0.32]{\includegraphics{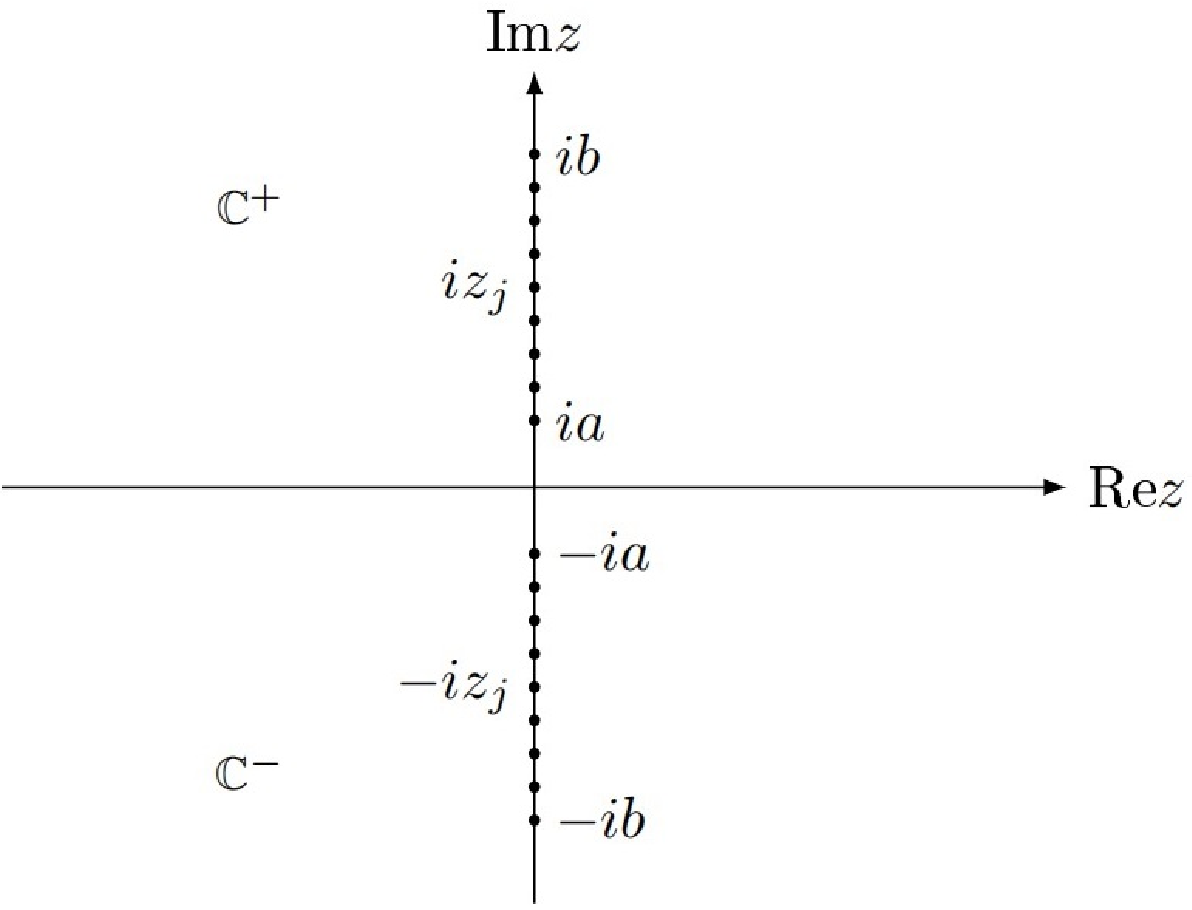}}}\hspace{-0.35in}
\vspace{0.05in}
\caption{Analytical domains of eigenfunctions and distribution of the discrete spectrum $Z_N$.}
   \label{fig8}
\end{figure}

Recently, Girotti {\it et al}~\cite{Girotti-1,Girotti-2} investigated the soliton gas phenomena of KdV and mKdV equations with discrete spectra located on the imaginary axis. Similarly, we assume that the discrete spectrum belongs to the interval $(ia, ib)\cup(-ib,-ia)$. We mainly focus on the situation where the limit $N\to\infty$ of $q_N(x,t)$, under the following conditions~\cite{Girotti-1}:

\begin{itemize}

 \item {} The simple poles $\{iz_j\}_{j=1}^N$ are sampled from a smooth density function $p(z)$ so that $\int_{a}^{z_j}p(\lambda)d\lambda=\frac{j}{N},j=1,\cdots,N$ (see Fig.~\ref{fig8}).

 \item {} The norming parameters $\{c_j\}_{j=1}^N$ satisfy
\bee
c_j=\dfrac{i(b-a)r(iz_j)}{N\pi},
\ene
where $r(z)$ is a real-valued, continuous and non-vanishing function for $z\in(ia,ib)\, (0<a<b)$, with the symmetry $r(z^*)=r(z)$.

\end{itemize}

Then the residue conditions (\ref{Res1}) of matrix function $M(x,t;z)$ can be rewrite as:
\begin{align}\label{Res10}
\begin{aligned}
&\mathrm{Res}_{z=iz_j}M(x,t;z)=\lim\limits_{z\to iz_j}M(x,t;z)\left[\!\!\begin{array}{cc}
0& 0  \vspace{0.05in}\\
\dfrac{i(b-a)r(iz_j)}{N\pi}e^{-2z_jx+(8\beta z_j^3-4i\alpha z_j^2)t}& 0
\end{array}\!\!\right],\vspace{0.05in}\\
&\mathrm{Res}_{z=-iz_j}M(x,t;z)=\lim\limits_{z\to -iz_j}M(x,t;z)\left[\!\!\begin{array}{cc}
0& \dfrac{i(b-a)r(iz_j)}{N\pi}e^{-2z_jx+(8\beta z_j^3+4i\alpha z_j^2)t}  \vspace{0.05in}\\
0& 0
\end{array}\!\!\right].
\end{aligned}
\end{align}
Notice that for the fixed $t=t_0$, as $x\to +\infty$ the residue conditions (\ref{Res10}) only have the exponentially small terms such that
the potential $q$ is exponentially small by a small norm argument. However, as $x\to -\infty$
both terms of the residue conditions (\ref{Res10}) are exponentially large. For this case, one may make the transform
 \bee
  \widehat M(z)=M(z)T(z),\quad T(z)=\prod_{j=1}^N\left(\frac{z-iz_j}{z+iz_j}\right)^{\sigma_3}
 \ene
such that the residue conditions of $\widehat M(z)$ are
\begin{align}\label{Res10g}
\begin{aligned}
&\mathrm{Res}_{z=iz_j}\widehat M(x,t;z)=\lim\limits_{z\to iz_j}M(x,t;z)\left[\!\!\begin{array}{cc}
0& \dfrac{N\pi}{i(b-a)r(iz_j)T'^2(iz_j)}e^{2z_jx-(8\beta z_j^3-4i\alpha z_j^2)t} \vspace{0.05in}\\
0 & 0
\end{array}\!\!\right],\vspace{0.05in}\\
&\mathrm{Res}_{z=-iz_j}\widehat M(x,t;z)=\lim\limits_{z\to -iz_j}M(x,t;z)\left[\!\!\begin{array}{cc}
0& 0 \vspace{0.05in}\\
\dfrac{N\pi T^4(-iz_j)}{i(b-a)r(iz_j)T'^2(-iz_j)}e^{2z_jx-(8\beta z_j^3+4i\alpha z_j^2)t}  & 0
\end{array}\!\!\right],
\end{aligned}
\end{align}
in which the moduli of  coefficients of two exponential functions are less than $e^{c_1N+c_2}$ for some positive constans $c_{1,2}>0$ such that one knows that for the fixed $t=t_0$, these two non-zero terms are exponentially small as $x\to -\infty$ with $x < (c_1N+c_2)/(c_3-2{\rm min}\{z_j\})$ for some positive constant $c_3$ ($0<c_3<2z_j$).  The corresponding potential is
\bee\label{fanyan-1g}
q_N(x,t)= 2i\lim\limits_{z\rightarrow\infty}(z \widehat M(x,t;z))_{12},
\ene

Let a closed curve $\Gamma_+$ ($\Gamma_-$ ) be a small radius encircling the poles $\{iz_j\}_{j=1}^N$ ($\{-iz_j\}_{j=1}^N$) counterclockwise (clockwise) in the upper (lower) half plane $\mathbb{C}_+$ ($\mathbb{C}_-$). Make the following transformation:
\bee
M^{(1)}(x,t;z)
=\begin{cases}
M(x,t;z)\left[\!\!\begin{array}{cc}
1& 0  \vspace{0.05in}\\
-\sum\limits_{j=1}^{N}\dfrac{i(b-a)r(iz_j)}{N\pi(z-iz_j)}e^{2i(z_jx+2\alpha z_j^2t+4\beta z_j^3t)} & 1
\end{array}\!\!\right],\quad z~\mathrm{within}~\Gamma_+,\vspace{0.05in}\\
M(x,t;z)\left[\!\!\begin{array}{cc}
1& -\sum\limits_{j=1}^{N}\dfrac{i(b-a)r(iz_j)}{N\pi(z+iz_j)}e^{-2i(z_jx+2\alpha z_j^2t+4\beta z_j^3t)}  \vspace{0.05in}\\
0& 1
\end{array}\!\!\right],\quad z~\mathrm{within}~\Gamma_-,\v\\
M(x,t;z),\quad \mathrm{otherwise}.
\end{cases}
\ene
Then the matrix function $M^{(1)}(x,t;z)$ satisfies the following Riemann-Hilbert problem.

\begin{prop}\label{RH2} 
Find a $2\times 2$ matrix function $M^{(1)}(x,t;z)$ that satisfies the following properties:

\begin{itemize}

 \item {} Analyticity: $M^{(1)}(x,t;z)$ is analytic in $\mathbb{C}\setminus(\Gamma_+\cup\Gamma_-)$ and takes continuous boundary values on $\Gamma_+\cup\Gamma_-$.

 \item {} Jump condition: The boundary values on the jump contour $\Gamma_+\cup\Gamma_-$ are defined as
 \bee\label{V1}
M_+^{(1)}(x,t;z)=M_-^{(1)}(x,t;z)V_1(x,t;z),\quad z\in\Gamma_+\cup\Gamma_-,
\ene
where
\bee\label{V1-1}
V_1(x,t;z)
=\begin{cases}
\left[\!\!\begin{array}{cc}
1& 0  \vspace{0.05in}\\
-\sum\limits_{j=1}^{N}\dfrac{i(b-a)r(iz_j)}{N\pi(z-iz_j)}e^{2i(z_jx+2\alpha z_j^2t+4\beta z_j^3t)} & 1
\end{array}\!\!\right],\quad z\in\Gamma_+,\vspace{0.05in}\\
\left[\!\!\begin{array}{cc}
1& \sum\limits_{j=1}^{N}\dfrac{i(b-a)r(iz_j)}{N\pi(z+iz_j)}e^{-2i(z_jx+2\alpha z_j^2t+4\beta z_j^3t)}  \vspace{0.05in}\\
0& 1
\end{array}\!\!\right],\quad z\in\Gamma_-;
\end{cases}
\ene

 \item {} Normalization: $M^{(1)}(x,t;z)=\mathbb{I}+\mathcal{O}(z^{-1}),\quad z\rightarrow\infty.$

\end{itemize}
\end{prop}

According to Eq.~(\ref{fanyan-1}), we can recover $q(x,t)$ by the following formula:
\bee\label{fanyan-2}
q(x,t)=2i\lim\limits_{z\rightarrow\infty}(z M^{(1)}(x,t;z))_{12}.
\ene

For convenience, we can assume that the $N$ poles are equally spaced along $(ia,ib)$ with distance between two poles equal to $s_d=\frac{b-a}{N}$.

Based on Ref.~\cite{Girotti-1}, one has the following lemma:

\begin{lemma}\label{le1}
For any open set $A_+$ {\rm (}$A_-${\rm)} containing the interval $[ia,ib]$ {\rm (}$[-ib,-ia]${\rm )}, the following identities hold:
\bee
\lim\limits_{N\to\infty}\sum\limits_{j=1}^{N}\frac{c_je^{2i(z_jx+2\alpha z_j^2t+4\beta z_j^3t)}}{z \mp iz_j}=
\pm \frac{1}{\pi}\int_{\pm ia}^{\pm ib}\dfrac{r(k)e^{2i(kx+2\alpha k^2t+4\beta k^3t)}}{z-k}dk,& z\in \mathbb{C}\setminus A_{\pm},
\ene
where the open intervals $(ia,ib)$ and $(-ib,-ia)$ are both oriented upwards.
\end{lemma}

\begin{proof}
For $z\in \mathbb{C}\setminus A_{+}$, since $c_j=\dfrac{i(b-a)r(iz_j)}{N\pi}$, one has
\begin{align}\no
\begin{array}{rl}
\d\lim\limits_{N\to\infty}\sum\limits_{j=1}^{N}\frac{c_je^{2i(z_jx+2\alpha z_j^2t+4\beta z_j^3t)}}{z-iz_j}&=\d\lim\limits_{N\to\infty}\sum\limits_{j=1}^{N}\frac{r(iz_j)(ib-ia)}{N\pi(z-iz_j)}e^{2i(z_jx+2\alpha z_j^2t+4\beta z_j^3t)} \v\\
&=\d\lim\limits_{N\to\infty}\sum\limits_{j=1}^{N}\frac{ir(iz_j)s_d}{\pi(z-iz_j)}e^{2i(z_jx+2\alpha z_j^2t+4\beta z_j^3t)}\v\\
&=\d \dfrac{1}{\pi}\int_{ia}^{ib}\frac{r(k)}{z-k}e^{2i(kx+2\alpha k^2t+4\beta k^3t)}dk.
\end{array}
\end{align}
Similarly, it can be proven that another identity hold $z\in \mathbb{C}\setminus A_{-}$. Thus the proof is completed.
\end{proof}

When $N\to\infty$, according to Lemma \ref{le1}, the jump matrix $V_1(x,t;z)$ defined by Eq.~(\ref{V1-1}) can be rewrite as:
\bee\label{V1-2}
V_1(x,t;z)
=\begin{cases}
\left[\!\!\begin{array}{cc}
1& 0  \vspace{0.05in}\\
-\dfrac{1}{\pi}\d\int_{ia}^{ib}\frac{r(k)}{z-k}e^{2i(kx+2\alpha k^2t+4\beta k^3t)}dk & 1
\end{array}\!\!\right],\quad z\in\Gamma_+,\vspace{0.1in}\\
\left[\!\!\begin{array}{cc}
1& \dfrac{1}{\pi}\d\int_{-ib}^{-ia}\frac{r(k)}{z-k}e^{-2i(kx+2\alpha k^2t+4\beta k^3t)}dk \vspace{0.05in}\\
0& 1
\end{array}\!\!\right],\quad z\in\Gamma_-;
\end{cases}
\ene

Make the following transformation:
\bee
M^{(2)}(x,t;z)
=\begin{cases}
M^{(1)}(x,t;z)\left[\!\!\begin{array}{cc}
1& 0  \vspace{0.05in}\\
\dfrac{1}{\pi}\d\int_{ia}^{ib}\frac{r(k)}{z-k}e^{2i(kx+2\alpha k^2t+4\beta k^3t)}dk & 1
\end{array}\!\!\right],\quad z~\mathrm{within}~\Gamma_+,\vspace{0.05in}\\
M^{(1)}(x,t;z)\left[\!\!\begin{array}{cc}
1& \dfrac{1}{\pi}\d\int_{-ib}^{-ia}\frac{r(k)}{z-k}e^{2i(kx+2\alpha k^2t+4\beta k^3t)}dk  \vspace{0.05in}\\
0& 1
\end{array}\!\!\right],\quad z~\mathrm{within}~\Gamma_-,\v\\
M^{(1)}(x,t;z),\quad \mathrm{otherwise}.
\end{cases}
\ene
For convenience, let $\gamma_1:=(a,b),\gamma_2:=(-b,-a)$ and both line segments $\gamma_1$ and $\gamma_2$ are directed to the right. Then matrix function $M^{(2)}(x,t;z)$ satisfies the following Riemann-Hilbert problem.

\begin{prop}\label{RH3} 
Find a $2\times 2$ matrix function $M^{(2)}(x,t;z)$ that satisfies the following properties:

\begin{itemize}

 \item {} Analyticity: $M^{(2)}(x,t;z)$ is analytic in $\mathbb{C}\setminus(i\gamma_1\cup i\gamma_2)$ and takes continuous boundary values on $i\gamma_1\cup i\gamma_2$.

 \item {} Jump condition: The boundary values on the jump contour $i\gamma_1\cup i\gamma_2$ are defined as
 \bee\label{V2}
M_+^{(2)}(x,t;z)=M_-^{(2)}(x,t;z)V_2(x,t;z),\quad z\in i\gamma_1\cup i\gamma_2,
\ene
where
\bee
V_2(x,t;z)
=\begin{cases}
\left[\!\!\begin{array}{cc}
1& 0  \vspace{0.05in}\\
2ir(z)e^{2i(xz+2\alpha tz^2+4\beta tz^3)}& 1
\end{array}\!\!\right],\quad z\in i\gamma_1,\vspace{0.05in}\\
\left[\!\!\begin{array}{cc}
1& 2ir(z)e^{-2i(xz+2\alpha tz^2+4\beta tz^3)}  \vspace{0.05in}\\
0& 1
\end{array}\!\!\right],\quad z\in i\gamma_2;
\end{cases}
\ene

 \item {} Normalization: $M^{(2)}(x,t;z)=\mathbb{I}+\mathcal{O}(z^{-1}),\quad z\rightarrow\infty.$

\end{itemize}
\end{prop}

According to Eq.~(\ref{fanyan-2}), we can recover $q(x,t)$ by the following formula:
\bee\label{fanyan-3}
q(x,t)=2i\lim\limits_{z\rightarrow\infty}(z M^{(2)}(x,t;z))_{12}.
\ene

We want to rotate the jump line to the real axis, so we will rewrite Eq.~(\ref{V2}) as follows:
\bee
M_+^{(2)}(x,t;iz)=M_-^{(2)}(x,t;iz)\begin{cases}
\left[\!\!\begin{array}{cc}
1& 0  \vspace{0.05in}\\
2ir(iz)e^{-2zx-4i\alpha z^2t+8\beta z^3t}& 1
\end{array}\!\!\right],\quad z\in \gamma_1,\vspace{0.05in}\\
\left[\!\!\begin{array}{cc}
1& 2ir(iz)e^{2zx+4i\alpha z^2t-8\beta z^3t}  \vspace{0.05in}\\
0& 1
\end{array}\!\!\right],\quad z\in \gamma_2.
\end{cases}
\ene

Make the following transformation:
\bee
M^{(3)}(x,t;z)=M^{(2)}(x,t;iz),\quad R(z)=2r(iz),
\ene
Then matrix function $M^{(3)}(x,t;z)$ satisfies the following Riemann-Hilbert problem:

\v
\begin{prop}\label{RH4} 
Find a $2\times 2$ matrix function $M^{(3)}(x,t;z)$ that satisfies the following properties:

\begin{itemize}

 \item {} Analyticity: $M^{(3)}(x,t;z)$ is analytic in $\mathbb{C}\setminus(\gamma_1\cup \gamma_2)$ and takes continuous boundary values on $\gamma_1\cup \gamma_2$.

 \item {} Jump condition: The boundary values on the jump contour $\gamma_1\cup\gamma_2$ are defined as
 \bee
M_+^{(3)}(x,t;z)=M_-^{(3)}(x,t;z)V_3(x,t;z),\quad z\in \gamma_1\cup \gamma_2,
\ene
where
\bee
V_3(x,t;z)
=\begin{cases}
\left[\!\!\begin{array}{cc}
1& 0  \vspace{0.05in}\\
iR(z)e^{-2zx-4i\alpha z^2t+8\beta z^3t}& 1
\end{array}\!\!\right],\quad z\in \gamma_1,\vspace{0.05in}\\
\left[\!\!\begin{array}{cc}
1& iR(z)e^{2zx+4i\alpha z^2t-8\beta z^3t}  \vspace{0.05in}\\
0& 1
\end{array}\!\!\right],\quad z\in \gamma_2.
\end{cases}
\ene

 \item {} Normalization: $M^{(3)}(x,t;z)=\mathbb{I}+\mathcal{O}(z^{-1}),\quad z\rightarrow\infty$.

\end{itemize}
\end{prop}

According to Eq.~(\ref{fanyan-3}), we can recover $q(x,t)$ by the following formula:
\bee\label{fanyan-4}
q(x,t)=-2\lim\limits_{z\rightarrow\infty}(z M^{(3)}(x,t;z))_{12}.
\ene

\section{Large-space asymptotics of the potential $q(x,0)$ as $x\to-\infty$}

Recently, Girotti {\it et al}~\cite{Girotti-1} studied the asymptotic behavior of the soliton gas of KdV equation when $t=0$. Similarly, we will study the asymptotic behavior of soliton gas $q(x):=q(x,0)$ of the Hirota equation at $t=0$. Let $Y(x; z)=M^{(3)}(x,0;z)$. Then we construct the following Riemann-Hilbert problem:

\begin{prop}\label{RH5} 
Find a $2\times 2$ matrix function $Y(x;z)$ that satisfies the following properties:

\begin{itemize}

 \item {} Analyticity: $Y(x;z)$ is analytic in $\mathbb{C}\setminus(\gamma_1\cup \gamma_2)$ and takes continuous boundary values on $\gamma_1\cup \gamma_2$.

 \item {} Jump condition: The boundary values on the jump contour $\gamma_1\cup\gamma_2$ are defined as
 \bee
Y_+(x;z)=Y_-(x;z)V_4(x;z),\quad z\in \gamma_1\cup \gamma_2,
\ene
where
\bee
V_4(x;z)
=\begin{cases}
\left[\!\!\begin{array}{cc}
1& 0  \vspace{0.05in}\\
iR(z)e^{-2zx}& 1
\end{array}\!\!\right],\quad z\in \gamma_1,\vspace{0.05in}\\
\left[\!\!\begin{array}{cc}
1& iR(z)e^{2zx}  \vspace{0.05in}\\
0& 1
\end{array}\!\!\right],\quad z\in \gamma_2;
\end{cases}
\ene

 \item {} Normalization: $Y(x;z)=\mathbb{I}+\mathcal{O}(z^{-1}),\quad z\rightarrow\infty.$

\end{itemize}
\end{prop}

According to Eq.~(\ref{fanyan-4}), we can recover $q(x)$ by the following formula:
\bee\label{fanyan-5}
q(x)=-2\lim\limits_{z\rightarrow\infty}(z Y(x;z))_{12}.
\ene


To solve the Riemann-Hilbert problem \ref{RH5}, we make the following transformation:
\bee\label{Y1-Y}
Y_1(x;z)=Y(x;z)[f(z)e^{xg(z)}]^{\sigma_3},
\ene
where $f(z)$ and $g(z)$ are scalar functions which satisfy the following properties~\cite{Girotti-1}:
\begin{itemize}

 \item {} Analyticity: $f(z),\, g(z)$ are analytic in $\mathbb{C}\setminus(-b,b)$ and takes continuous boundary values on $(-b,b)$.

 \item {} Jump condition: The boundary values of $f(z),\, g(z)$ on the jump contour are defined as
\begin{align}
\begin{aligned}
&f_+(z)f_-(z)=\left\{\begin{array}{ll} R^{-1}(z), & z\in \gamma_1,\v\\
 R(z), & z\in \gamma_2,
 \end{array}\right.
 \v\\
&f_+(z)=f_-(z)e^{m_1},\quad z\in[-a,a],
\end{aligned}
\end{align}
and
\begin{align}
\begin{aligned}
&g_+(z)+g_-(z)=2z,\quad z\in \gamma_1\cup \gamma_2,\v\\
&g_+(z)-g_-(z)=m_2,\quad z\in [-a, a],
\end{aligned}
\end{align}
where
where
\begin{align}\no
\begin{aligned}
&m_1=-\frac{b}{K_1(m)}\int_{a}^{b}\dfrac{\ln R(\lambda)}{C_+(\lambda)}d\lambda,\quad C(z)=\sqrt{(z^2-a^2)(z^2-b^2)},\quad
m_2=-\dfrac{i\pi b}{K_1(m)},
\end{aligned}
\end{align}
and the first and second kinds of complete elliptic integrals are denoted as:
\bee
K_1(m):=\int_{0}^{1}\dfrac{ds}{\sqrt{(1-s^2)(1-m^2s^2)}}, \quad K_2(m):=\int_{0}^{1}\frac{\sqrt{1-m^2s^2}}{\sqrt{1-t^2}}ds.
\ene
with modulus $m=\frac{a}{b}$.

 \item {} Normalization: $f(z)=1+\mathcal{O}(z^{-1}),\,\, g(z)=\mathcal{O}(z^{-1}),\quad z\rightarrow\infty$.
\end{itemize}

The scalar functions $f(z)$ and $g(z)$ can also be accurately solved as~\cite{Girotti-1,Girotti-2}:
\begin{align}
\begin{aligned}
&f(z)=\exp\left[\frac{C(z)}{2\pi i}\left(\int_{\gamma_2-\gamma_1}\frac{\ln R(\lambda)}{C_+(\lambda)(\lambda-z)}d\lambda
+\int_{-a}^{a}\frac{m_1}{C(\lambda)(\lambda-z)}d\lambda
\right)\right],\v\\ \\
&g(z)=z-\int_{b}^{z}\frac{\lambda^2+b^2(K_2(m)/K_1(m)-1)}{C(\lambda)}d\lambda.
\end{aligned}
\end{align}

Then matrix function $Y_1(x;z)$ satisfies the following Riemann-Hilbert problem.

\begin{prop}\label{RH6} 
Find a $2\times 2$ matrix function $Y_1(x;z)$ that satisfies the following properties:

\begin{itemize}

 \item {} Analyticity: $Y_1(x;z)$ is analytic in $\mathbb{C}\setminus(-b,b)$ and takes continuous boundary values on $(-b,b)$.

 \item {} Jump condition: The boundary values on the jump contour are defined as
 \bee
Y_{1+}(x;z)=Y_{1-}(x;z)V_5(x;z),\quad z\in \gamma_1\cup[-a,a]\cup\gamma_2,
\ene
where
\bee
V_5(x;z)
=\left\{\begin{array}{ll}
\left[\!\!\begin{array}{cc}
\delta f(z) e^{x\Delta g(z)} & 0  \vspace{0.05in}\\
i& \delta^{-1}f(z)e^{-x\Delta g(z)}
\end{array}\!\!\right],& z\in \gamma_1,\vspace{0.05in}\\
\left[\!\!\begin{array}{cc}
e^{xm_2+m_1}& 0  \vspace{0.05in}\\
0& e^{-(xm_2+m_1)}
\end{array}\!\!\right],& z\in [-a,a],\vspace{0.05in}\\
\left[\!\!\begin{array}{cc}
\delta f(z) e^{x\Delta g(z)} & i  \vspace{0.05in}\\
0& \delta^{-1}f(z)e^{-x\Delta g(z)}
\end{array}\!\!\right],& z\in \gamma_2;
\end{array}\right.
\ene
with $\delta f(z)=f_+(z)/f_-(z),\,\Delta g(z)=g_+(z)-g_-(z).$

 \item {} Normalization: $Y_1(x;z)=\mathbb{I}+\mathcal{O}(z^{-1}),\quad z\rightarrow\infty.$
\end{itemize}
\end{prop}

\subsection{Opening lenses: jump matrix decompositions}

Based on the idea of Ref.~\cite{Girotti-1}, let us open the lens $O_1$ pass through points $z=a$ and $z=b$ and open the lens $O_2$ pass through points $z=-b$ and $z=-a$ (see Fig.~\ref{fig1}). We define a new function as follows:
\bee
R_{\pm}(z):=\pm R(z),\quad z\in\gamma_1\cup\gamma_2.
\ene

\begin{figure}[!t]
    \centering
 \vspace{-0.15in}
  {\scalebox{0.65}[0.65]{\includegraphics{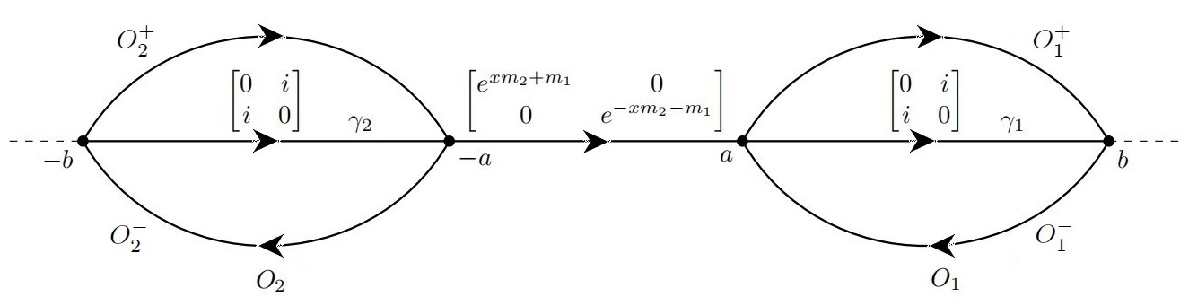}}}\hspace{-0.35in}
\vspace{0.05in}
\caption{The Riemann-Hilbert problem~\ref{RH7} for matrix function $Y_2(x; z)$. Opening lenses $O_1$ and $O_2$.}
   \label{fig1}
\end{figure}

Note that the jump matrix $V_5(x;z)|_{\gamma_1}$ has the following decomposition:
\begin{align}
\begin{aligned}
V_5(x;z)&=\left[\!\!\begin{array}{cc}
\delta f(z) e^{x\Delta g(z)} & 0  \vspace{0.05in}\\
i& \delta^{-1}f(z)e^{-x\Delta g(z)}
\end{array}\!\!\right] \vspace{0.05in}\\
&=\left[\!\!\begin{array}{cc}
1& \dfrac{ie^{x\Delta g(z)}}{R_{-}(z)f_-^2(z)}  \vspace{0.05in}\\
0& 1
\end{array}\!\!\right]\left[\!\!\begin{array}{cc}
0& i  \vspace{0.05in}\\
i& 0
\end{array}\!\!\right]\left[\!\!\begin{array}{cc}
1& -\dfrac{ie^{-x\Delta g(z)}}{R_{+}(z)f_+^2(z)}  \vspace{0.05in}\\
0& 1
\end{array}\!\!\right],\quad z\in \gamma_1,
\end{aligned}
\end{align}
and the jump matrix $V_5(x;z)|_{\gamma_2}$ has the following decomposition:
\begin{align}
\begin{aligned}
V_5(x;z)&=
\left[\!\!\begin{array}{cc}
\delta f(z) e^{x\Delta g(z)} & i  \vspace{0.05in}\\
0 & \delta^{-1}f(z)e^{-x\Delta g(z)}
\end{array}\!\!\right]\vspace{0.05in}\\
&=\left[\!\!\begin{array}{cc}
1& 0  \vspace{0.05in}\\
\dfrac{if_-^2(z)}{R_{-}(z)}e^{-x\Delta g(z)}& 1
\end{array}\!\!\right]\left[\!\!\begin{array}{cc}
0& i  \vspace{0.05in}\\
i& 0
\end{array}\!\!\right]\left[\!\!\begin{array}{cc}
1& 0  \vspace{0.05in}\\
-\dfrac{if_+^2(z)}{R_{+}(z)}e^{x\Delta g(z)}& 1
\end{array}\!\!\right],\quad z\in \gamma_2.
\end{aligned}
\end{align}

According to the above-mentioned  decompositions with opening lenses, we make the following transformation:
\bee\label{Y2-Y1}
Y_2(x;z)=\begin{cases}
Y_1(x;z)\left[\!\!\begin{array}{cc}
1& \dfrac{ie^{-2x(g(z)-z)}}{R(z)f^2(z)}  \vspace{0.05in}\\
0& 1
\end{array}\!\!\right],\qquad\quad z~\mathrm{in~the~upper}~(O_1^+)/\mathrm{lower} (O_1^-)~\mathrm{part~of~lens}~O_1,\vspace{0.05in}\\
Y_1(x;z)\left[\!\!\begin{array}{cc}
1& 0  \vspace{0.05in}\\
\dfrac{if^2(z)}{R(z)}e^{2x(g(z)-z)}& 1
\end{array}\!\!\right],\quad z~\mathrm{in~the~upper}~(O_2^+)/\mathrm{lower} (O_2^-)~\mathrm{part~of~lens}~O_2,\vspace{0.05in}\\
Y_1(x;z),\quad \mathrm{otherwise}.
\end{cases}
\ene

Then the matrix function $Y_2(x;z)$ satisfies the following Riemann-Hilbert problem.

\begin{prop}\label{RH7} 
Find a $2\times 2$ matrix function $Y_2(x;z)$ that satisfies the following properties:

\begin{itemize}

 \item {} Analyticity: $Y_2(x;z)$ is analytic in $\mathbb{C}\setminus((-b,b)\cup O_1\cup O_2)$ and takes continuous boundary values on $(-b,b)\cup O_1\cup O_2$.

 \item {} Jump condition: The boundary values on the jump contour are defined as
 \bee
Y_{2+}(x;z)=Y_{2-}(x;z)V_6(x;z),\quad z\in \gamma_1\cup[-a,a]\cup\gamma_2\cup O_1^+\cup O_1^-\cup O_2^+\cup O_2^-,
\ene
where the jump matrix is defined as
\bee\label{V6}
V_6(x;z)
=\begin{cases}
\left[\!\!\begin{array}{cc}
1& -\dfrac{ie^{-2x(g(z)-z)}}{R(z)f^2(z)}  \vspace{0.05in}\\
0& 1
\end{array}\!\!\right],\quad z\in O_1^+\cup O_1^-,\vspace{0.05in}\\
\left[\!\!\begin{array}{cc}
1& 0  \vspace{0.05in}\\
-\dfrac{if^2(z)}{R(z)}e^{2x(g(z)-z)} & 1
\end{array}\!\!\right],\quad z\in O_2^+\cup O_2^-,\v\\
\left[\!\!\begin{array}{cc}
e^{xm_2+m_1}& 0  \vspace{0.05in}\\
0& e^{-(xm_2+m_1)}
\end{array}\!\!\right],\quad z\in [-a,a],\vspace{0.05in}\\
\left[\!\!\begin{array}{cc}
0& i  \vspace{0.05in}\\
i& 0
\end{array}\!\!\right],\quad z\in  \gamma_1\cup\gamma_2;
\end{cases}
\ene
 \item {} Normalization: $Y_2(x;z)=\mathbb{I}+\mathcal{O}(z^{-1}),\quad z\rightarrow\infty.$
\end{itemize}
\end{prop}

\begin{lemma}~\cite{Girotti-1}\label{le2} The following inequalities hold:
\begin{align}
\begin{aligned}
&\mathrm{Re}(g(z)-z)<0,\quad z\in O_1\setminus\{a,b\},\v\\
&\mathrm{Re}(z-g(z))<0,\quad z\in O_2\setminus\{-a,-b\}.
\end{aligned}
\end{align}

\end{lemma}

According to Lemma \ref{le2}, we know that the off-diagonal entries of the jump matrix defined by Eq.~(\ref{V6}) along the
upper and lower lenses $O_1,O_2$ are exponentially decay when $x\to-\infty$.

\subsection{The outer matrix  parametrix}

We construct a new Riemann-Hilbert problem with the matrix function having only jumps on $(-b,b)$.

\begin{prop}\label{RH8} 
Find a $2\times 2$ matrix function $Y^{o}(x;z)$ that satisfies the following properties:

\begin{itemize}

 \item {} Analyticity: $Y^o(x;z)$ is analytic in $\mathbb{C}\setminus(-b,b)$ and takes continuous boundary values on $(-b,b)$.

 \item {} Jump condition: The boundary values on the jump contour are defined as
 \bee
Y^o_{+}(x;z)=Y^o_{-}(x;z)V_7(x;z),\quad z\in \gamma_1\cup[-a,a]\cup\gamma_2,
\ene
where
\bee
V_7(x;z)
=\begin{cases}
\left[\!\!\begin{array}{cc}
0& i  \vspace{0.05in}\\
i& 0
\end{array}\!\!\right],\quad z\in \gamma_1\cup \gamma_2,\vspace{0.05in}\\
\left[\!\!\begin{array}{cc}
e^{xm_2+m_1}& 0  \vspace{0.05in}\\
0& e^{-(xm_2+m_1)}
\end{array}\!\!\right],\quad z\in [-a,a].
\end{cases}
\ene

 \item {} Normalization: $Y^o(x;z)=\mathbb{I}+\mathcal{O}(z^{-1}),\quad z\rightarrow\infty.$
\end{itemize}
\end{prop}

Notice that the Riemann-Hilbert problem \ref{RH8} is similar to the one arising from the long-time asymptotic behavior of KdV equation with step-like initial conditions~\cite{kdv-13}.
In order to solve the Riemann-Hilbert problem \ref{RH8}, we will define a two-sheeted Riemann surface $\Omega$ of genus $1$ related to the function $C(z)$ (see Fig.~\ref{fig2}),
\bee\label{Omega}
\Omega:=\left\{(\zeta_1,\zeta_2)\in\mathbb{C}^2|\zeta_2^2=C^2(z)=(z^2-a^2)(z^2-b^2)\right\}.
\ene

\begin{figure}[!t]
    \centering
 \vspace{-0.15in}
  {\scalebox{0.55}[0.55]{\includegraphics{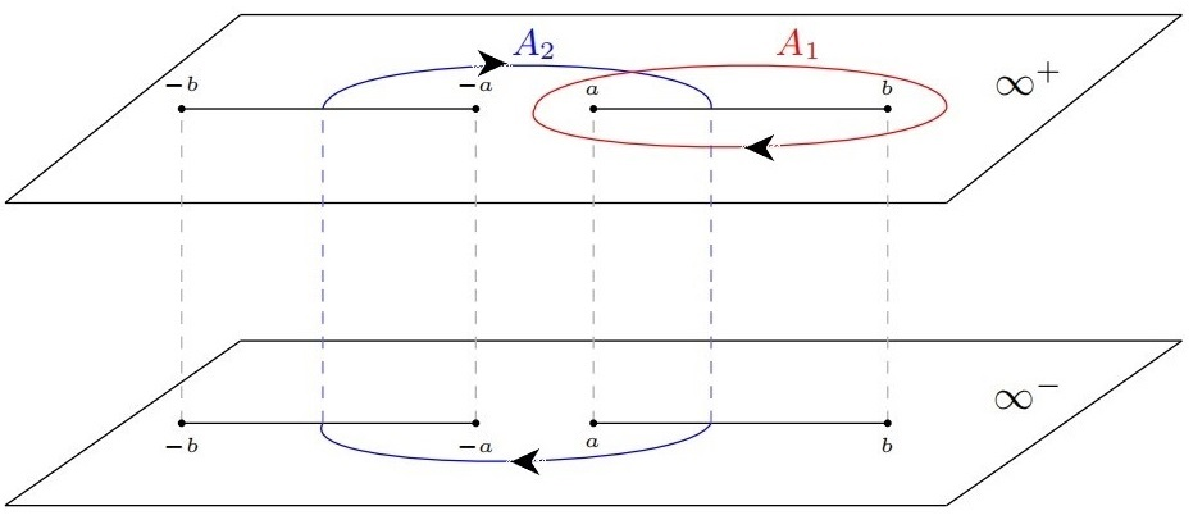}}}\hspace{-0.35in}
\vspace{0.1in}
\caption{The two-sheeted Riemann surface with genus 1, $\Omega$, given by Eq.~(\ref{Omega}).}
   \label{fig2}
\end{figure}

Note that
\begin{align}
\begin{aligned}
&\oint_{A_1}\frac{m_2}{4\pi iC(\lambda)}d\lambda=\dfrac{iK_1(\sqrt{1-m^2})}{2K_1(m)}=:\tau,\v\\
&\oint_{A_2}\frac{m_2}{4\pi iC(\lambda)}d\lambda=1.
\end{aligned}
\end{align}

Let
\begin{align}
\begin{aligned}
&m_3(z)=\int_{b}^{z}\frac{m_2}{4\pi iC(\lambda)}d\lambda,\v\\
&m_4(z)=\int_{b}^{z}\frac{\lambda^2+b^2(\frac{K_2(m)}{K_1(m)}-1)}{C(\lambda)}d\lambda,\v\\
&m_5(z)=\left(\frac{(z-b)(z+a)}{(z-a)(z+b)}\right)^{\frac14},
\end{aligned}
\end{align}
Then one has the following jump relations:
\begin{align}
\begin{aligned}
&m_{3+}(z)+m_{3-}(z)=0,\quad z\in\gamma_1,\v\\
&m_{3+}(z)-m_{3-}(z)=-\tau,\quad z\in(-a,a),\v\\
&m_{3+}(z)+m_{3-}(z)=-1,\quad z\in\gamma_2,\v\\
&m_{3+}(z)-m_{3-}(z)=0,\quad z\in[b,+\infty),\v\\
&m_{5+}(z)=im_{5-}(z),\quad z\in\gamma_1\cup\gamma_2,
\end{aligned}
\end{align}

Based on Ref.~\cite{Girotti-1}, one has the following lemma:

\begin{lemma}\label{le4} These functions:
\begin{align}
\begin{aligned}
&\psi_{11}(x;z)=(m_5(z)+m_5^{-1}(z))\dfrac{\vartheta_3(0;\tau)\vartheta_3(m_3(z)+\frac{xm_2+m_1}{2\pi i}+\frac14;\tau)}{\vartheta_3(m_3(z)+\frac14;\tau)\vartheta_3(\frac{xm_2+m_1}{2\pi i};\tau)},\v\\
&\psi_{21}(x;z)=(m_5(z)-m_5^{-1}(z))\dfrac{\vartheta_3(0;\tau)\vartheta_3(m_3(z)+\frac{xm_2+m_1}{2\pi i}-\frac14;\tau)}{\vartheta_3(m_3(z)-\frac14;\tau)\vartheta_3(\frac{xm_2+m_1}{2\pi i};\tau)},\v\\
&\psi_{12}(x;z)=(m_5(z)-m_5^{-1}(z))\dfrac{\vartheta_3(0;\tau)\vartheta_3(-m_3(z)+\frac{xm_2+m_1}{2\pi i}+\frac14;\tau)}{\vartheta_3(-m_3(z)+\frac14;\tau)\vartheta_3(\frac{xm_2+m_1}{2\pi i};\tau)},\v\\
&\psi_{22}(x;z)=(m_5(z)+m_5^{-1}(z))\dfrac{\vartheta_3(0;\tau)\vartheta_3(-m_3(z)+\frac{xm_2+m_1}{2\pi i}-\frac14;\tau)}{\vartheta_3(-m_3(z)-\frac14;\tau)\vartheta_3(\frac{xm_2+m_1}{2\pi i};\tau)},
\end{aligned}
\end{align}
have the following jump relations:
\begin{align}
\begin{aligned}
&\psi_{11+}(x;z)=i\psi_{12-}(x;z),\quad \psi_{21+}(x;z)=i\psi_{22-}(x;z),\quad z\in\gamma_1\cup\gamma_2,\v\\
&\psi_{12+}(x;z)=i\psi_{11-}(x;z),\quad \psi_{22+}(x;z)=i\psi_{21-}(x;z),\quad z\in\gamma_1\cup\gamma_2,\v\\
&\psi_{11+}(x;z)=\psi_{11-}(x;z)e^{xm_2+m_1},\quad \psi_{12+}(x;z)=\psi_{12-}(x;z)e^{-(xm_2+m_1)},\quad z\in(-a,a),\v\\
&\psi_{21+}(x;z)=\psi_{21-}(x;z)e^{xm_2+m_1},\quad \psi_{22+}(x;z)=\psi_{22-}(x;z)e^{-(xm_2+m_1)},\quad z\in(-a,a).
\end{aligned}
\end{align}
where the third Jacobi theta function is defined as
\bee
\vartheta_3(z;\tau)=\sum_{n=-\infty}^{+\infty}e^{i(2\pi nz+\pi n^2\tau)},\quad z\in\mathbb{C},
\ene
which has the periodicity condition:
\bee\label{peri-cond}
\vartheta_3(z+h\tau+l;\tau)=e^{-(2\pi hz+\pi h^2\tau)}\vartheta_3(z;\tau),\quad h,\, l\in\mathbb{Z}.
\ene
\end{lemma}

\begin{proof}
If $z\in\gamma_1\cup\gamma_2$, we have
\begin{align}\no
\begin{aligned}
\psi_{11+}(x;z)&=(m_{5+}(z)+m_{5+}^{-1}(z))\dfrac{\vartheta_3(0;\tau)\vartheta_3(m_{3+}(z)+\frac{xm_2+m_1}{2\pi i}+\frac14;\tau)}{\vartheta_3(m_{3+}(z)+\frac14;\tau)\vartheta_3(\frac{xm_2+m_1}{2\pi i};\tau)}\v\\
&=i(m_{5-}(z)-m_{5-}^{-1}(z))\dfrac{\vartheta_3(0;\tau)\vartheta_3(-m_{3-}(z)+\frac{xm_2+m_1}{2\pi i}+\frac14;\tau)}{\vartheta_3(-m_{3-}(z)+\frac14;\tau)\vartheta_3(\frac{xm_2+m_1}{2\pi i};\tau)}\v\\
&=i\psi_{12-}(x;z).
\end{aligned}
\end{align}

If $z\in(-a,a)$, we have
\begin{align}\label{le4-1}
\begin{aligned}
\psi_{11+}(x;z)&=(m_{5}(z)+m_{5}^{-1}(z))\dfrac{\vartheta_3(0;\tau)\vartheta_3(m_{3+}(z)+\frac{xm_2+m_1}{2\pi i}+\frac14;\tau)}{\vartheta_3(m_{3+}(z)+\frac14;\tau)\vartheta_3(\frac{xm_2+m_1}{2\pi i};\tau)}\v\\
&=(m_{5}(z)+m_{5}^{-1}(z))\dfrac{\vartheta_3(0;\tau)\vartheta_3(m_{3-}(z)-\tau+\frac{xm_2+m_1}{2\pi i}+\frac14;\tau)}{\vartheta_3(m_{3-}(z)-\tau+\frac14;\tau)\vartheta_3(\frac{xm_2+m_1}{2\pi i};\tau)}.
\end{aligned}
\end{align}

According to Eq.~(\ref{peri-cond}), we have
\begin{align}\label{le4-2}
\begin{aligned}
&\dfrac{\vartheta_3(m_{3-}(x;z)-\tau+\frac{xm_2+m_1}{2\pi i}+\frac14;\tau)}{\vartheta_3(m_{3-}(z)-\tau+\frac14;\tau)}\v\\
=&\dfrac{e^{-2i\pi\tau+2\pi i(m_{3-}(z)+\frac{xm_2+m_1}{2\pi i}+\frac14)}\vartheta_3(m_{3-}(z)+\frac{xm_2+m_1}{2\pi i}+\frac14;\tau)}{e^{-2i\pi\tau+2\pi i(m_{3-}(z)+\frac14)}\vartheta_3(m_{3-}(z)-\tau+\frac14;\tau)}\v\\
=&\dfrac{\vartheta_3(m_{3-}(z)+\frac{xm_2+m_1}{2\pi i}+\frac14;\tau)}{\vartheta_3(m_{3-}(z)+\frac14;\tau)}e^{xm_2+m_1},
\end{aligned}
\end{align}

Substituting Eq.~(\ref{le4-2}) into Eq.~(\ref{le4-1}), we obtain
\bee
\psi_{11+}(x;z)=\psi_{11-}(x;z)e^{xm_2+m_1}.
\ene
Similarly, it can be proven that other equations also hold. Thus the proof is completed.
\end{proof}

\begin{propo} The solution of Riemann-Hilbert problem \ref{RH8} can be expressed in the following form.
\bee\label{Yoxz}
Y^o(x;z)=\frac12\left[\!\!\begin{array}{cc}
\psi_{11}(x;z)& \psi_{12}(x;z)  \vspace{0.05in}\\
\psi_{21}(x;z)& \psi_{22}(x;z)
\end{array}\!\!\right].
\ene

\end{propo}

\begin{proof} According to Lemma \ref{le4}, we know that the matrix function $Y^o(x;z)$ is a solution of Riemann-Hilbert problem \ref{RH8}.

\end{proof}

\subsection{The local matrix parametrix }

In this section, we will construct a local matrix parametrix $Y^{b}(x;z)$ as $z\in\delta^{b}:=\{z||z-b|<\epsilon,~\epsilon~\mathrm{is~a~suitable~small,~ positive~parameter} \}$. Similarly, we can define regions $\delta^{a}$, $\delta^{-b}$, and $\delta^{-a}$. Note that
\bee
z-g_{\pm}(z)=\mathcal{O}(\sqrt{z-b}),\quad z\to b.
\ene

Then, we define the following conformal map:
\bee
\zeta_b:=\dfrac{x^2(g(z)-z)^2}{4},\quad z\in\delta^{b}.
\ene

Make the following transformation:
\bee
Y^{(1)}(x;z)
=\begin{cases}
Y_2(x;z)\left(\dfrac{\exp(\frac{i\pi}{4})}{\sqrt{R(z)}f(z)}\right)^{\sigma_3}e^{-2\sqrt{\zeta_b}\sigma_3}\left[\!\!\begin{array}{cc}
0& 1 \vspace{0.05in}\\
1& 0
\end{array}\!\!\right],\quad z\in\mathbb{C}_+\cap\delta^{b},\v\v\\
Y_2(x;z)\left(\dfrac{\exp(\frac{i\pi}{4})}{\sqrt{-R(z)}f(z)}\right)^{\sigma_3}e^{-2\sqrt{\zeta_b}\sigma_3}\left[\!\!\begin{array}{cc}
0& 1 \vspace{0.05in}\\
1& 0
\end{array}\!\!\right],\quad z\in\mathbb{C}_-\cap\delta^{b}.
\end{cases}
\ene

\begin{lemma} The matrix function $Y^{(1)}(x;z)$ satisfies the following jump conditions:
\bee
Y_+^{(1)}(x;z)
=\begin{cases}
Y_-^{(1)}(x;z)\left[\!\!\begin{array}{cc}
1& 0 \vspace{0.05in}\\
-1& 1
\end{array}\!\!\right],\quad z\in\delta^{b}\cap\{\mathrm{upper~and~lower~lenses}\},\v\v\\
Y_-^{(1)}(x;z)\left[\!\!\begin{array}{cc}
0& -1 \vspace{0.05in}\\
1& 0
\end{array}\!\!\right],\quad z\in\delta^{b}\cap[0,+\infty).
\end{cases}
\ene

\end{lemma}

\begin{proof}
If $z\in\delta^{b}\cap[0,+\infty)$, we have

\begin{align}\no
\begin{aligned}
Y_+^{(1)}(x;z)&=Y_{2+}(x;z)\left(\dfrac{\exp(\frac{i\pi}{4})}{\sqrt{R(z)}f(z)}\right)^{\sigma_3}e^{-2\sqrt{\zeta_b}\sigma_3}\left[\!\!\begin{array}{cc}
0& 1 \vspace{0.05in}\\
1& 0
\end{array}\!\!\right]\v\\
&=Y_{2-}(x;z)\left[\!\!\begin{array}{cc}
0& i  \vspace{0.05in}\\
i& 0
\end{array}\!\!\right]\left(\dfrac{\exp(\frac{i\pi}{4})}{\sqrt{R(z)}f(z)}\right)^{\sigma_3}e^{-2\sqrt{\zeta_b}\sigma_3}\left[\!\!\begin{array}{cc}
0& 1 \vspace{0.05in}\\
1& 0
\end{array}\!\!\right]\v\\
&=Y_-^{(1)}(x;z)\left[\!\!\begin{array}{cc}
0& 1 \vspace{0.05in}\\
1& 0
\end{array}\!\!\right]e^{2\sqrt{\zeta_b}\sigma_3}\left(\dfrac{\exp(\frac{i\pi}{4})}{\sqrt{R(z)}f(z)}\right)^{-\sigma_3}\left[\!\!\begin{array}{cc}
0& i  \vspace{0.05in}\\
i& 0
\end{array}\!\!\right]\left(\dfrac{\exp(\frac{i\pi}{4})}{\sqrt{R(z)}f(z)}\right)^{\sigma_3}e^{-2\sqrt{\zeta_b}\sigma_3}\left[\!\!\begin{array}{cc}
0& 1 \vspace{0.05in}\\
1& 0
\end{array}\!\!\right]\v\\
&=Y_-^{(1)}(x;z)\left[\!\!\begin{array}{cc}
0& -1 \vspace{0.05in}\\
1& 0
\end{array}\!\!\right].
\end{aligned}
\end{align}

If $z\in\delta^{b}\cap\{\mathrm{upper~lens}\}$, we have
\begin{align}\no
\begin{aligned}
Y_+^{(1)}(x;z)&=Y_{2+}(x;z)\left(\dfrac{\exp(\frac{i\pi}{4})}{\sqrt{R(z)}f(z)}\right)^{\sigma_3}e^{-2\sqrt{\zeta_b}\sigma_3}\left[\!\!\begin{array}{cc}
0& 1 \vspace{0.05in}\\
1& 0
\end{array}\!\!\right]\v\\
&=Y_{2-}(x;z)\left[\!\!\begin{array}{cc}
1& -\dfrac{ie^{-2x(g(z)-z)}}{R(z)f^2(z)}  \vspace{0.05in}\\
0& 1
\end{array}\!\!\right]\left(\dfrac{\exp(\frac{i\pi}{4})}{\sqrt{R(z)}f(z)}\right)^{\sigma_3}e^{-2\sqrt{\zeta_b}\sigma_3}\left[\!\!\begin{array}{cc}
0& 1 \vspace{0.05in}\\
1& 0
\end{array}\!\!\right]\v\\
&=Y_-^{(1)}(x;z)\left[\!\!\begin{array}{cc}
0& 1 \vspace{0.05in}\\
1& 0
\end{array}\!\!\right]e^{2\sqrt{\zeta_b}\sigma_3}\left(\dfrac{\exp(\frac{i\pi}{4})}{\sqrt{R(z)}f(z)}\right)^{-\sigma_3}\left[\!\!\begin{array}{cc}
1& -\dfrac{ie^{-2x(g(z)-z)}}{R(z)f^2(z)}  \vspace{0.05in}\\
0& 1
\end{array}\!\!\right]\v\\
&\quad \times \left(\dfrac{\exp(\frac{i\pi}{4})}{\sqrt{R(z)}f(z)}\right)^{\sigma_3}e^{-2\sqrt{\zeta_b}\sigma_3}\left[\!\!\begin{array}{cc}
0& 1 \vspace{0.05in}\\
1& 0
\end{array}\!\!\right]\v\\
&=Y_-^{(1)}(x;z)\left[\!\!\begin{array}{cc}
1& 0 \vspace{0.05in}\\
-1& 1
\end{array}\!\!\right].
\end{aligned}
\end{align}

As the same way, we can obtain the jump condition of matrix function $Y^{(1)}(x;z)$ as $z\in\delta^{b}\cap\{\mathrm{lower~lens}\}$. Thus the proof is completed.
\end{proof}

\begin{figure}[!t]
    \centering
 \vspace{-0.15in}
  {\scalebox{0.26}[0.26]{\includegraphics{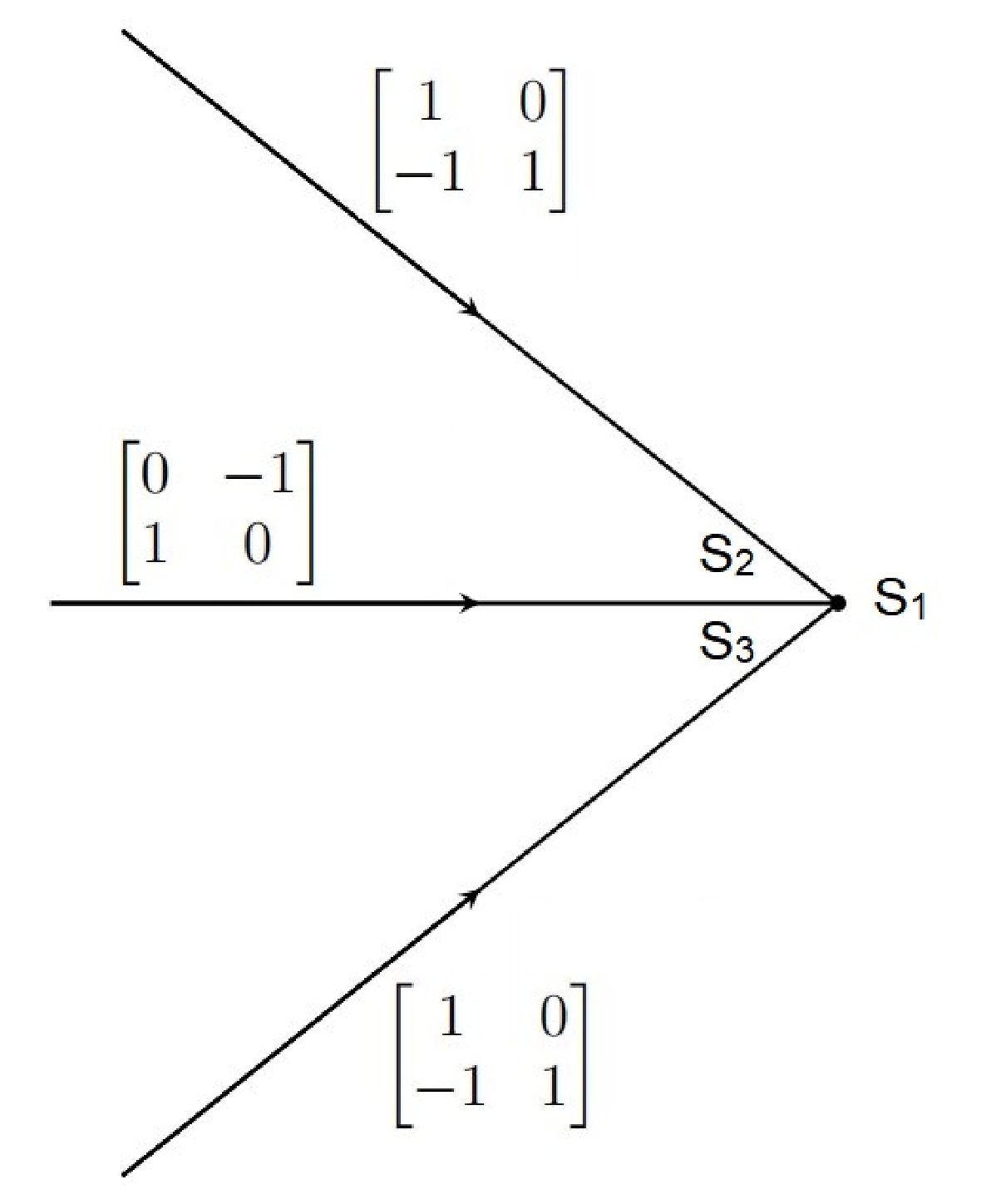}}}\hspace{-0.35in}
\vspace{0.05in}
\caption{Regional division for the Bessel model $U_{Bes}(z)$.}
   \label{fig5}
\end{figure}

Next, we will introduce the Bessel model $U_{Bes}(z)$ which satisfies the following Riemann-Hilbert problem~\cite{Girotti-1,Kuijlaars-1}:
\begin{prop}\label{BES}
Find a $2\times 2$ matrix function $U_{Bes}(z)$ that satisfies the following properties:

\begin{itemize}

 \item {} Analyticity: $U_{Bes}(z)$ is analytic for $z$ in the three regions shown in Figure \ref{fig5}, namely, $S_1:=\{z\in\mathbb{C}||\mathrm{arg}(z)|<\frac{2\pi}{3}\}$,
 $S_2:=\{z\in\mathbb{C}|\frac{2\pi}{3}<\mathrm{arg}(z)<\pi\}$ and $S_3:=\{z\in\mathbb{C}|-\pi<\mathrm{arg}(z)<-\frac{2\pi}{3}\}$, where $-\pi<\mathrm{arg}(z)\leq\pi$, and takes continuous boundary values on the excluded rays and at the origin from each sector.

 \item {} Jump condition: The boundary values on the jump contour $\gamma_U$ where $\gamma_U:=\gamma_{\pm}\cup\gamma_0$ with $\gamma_{\pm}=\{
 \mathrm{arg}(z)=\pm\frac{2\pi}{3}\}$ and $\gamma_0=\{
 \mathrm{arg}(z)=\pi\}$ are defined as
 \bee
U_{Bes+}(z)=U_{Bes-}(z)V_{Bes}(z),\quad z\in \gamma_U,
\ene
where
\bee
V_{Bes}(z)
=\begin{cases}
\left[\!\!\begin{array}{cc}
1& 0 \vspace{0.05in}\\
-1& 1
\end{array}\!\!\right],\quad z\in \gamma_{\pm},\vspace{0.05in}\\
\left[\!\!\begin{array}{cc}
0& -1 \vspace{0.05in}\\
1& 0
\end{array}\!\!\right],\quad z\in \gamma_0;
\end{cases}
\ene

 \item {} Normalization:
 \bee\no
U_{Bes}(z)=\mathcal{O}(\ln(|z|))\left[\!\!\begin{array}{cc}
1 & 1 \vspace{0.05in}\\ 1 &  1 \end{array}\!\!\right],\quad z\rightarrow 0.
\ene

\end{itemize}

\end{prop}

The solution of the Riemann-Hilbert Problem \ref{BES} can be expressed explicitly according to the Bessel functions~\cite{Kuijlaars-1}.  In particular, the solution $U_{Bes}(z)$ satisfies:
\bee
U_{Bes}(z)=\frac{\sqrt{2}}{2}(2\pi\sqrt{z})^{-\frac{\sigma_3}{2}}\left[\!\!\begin{array}{cc}
1& -i \vspace{0.05in}\\
-i& 1
\end{array}\!\!\right]\left(\mathbb{I}+\mathcal{O}(z^{-\frac12})\right)e^{2\sqrt{z}\sigma_3},\quad z\to \infty,
\ene
in $\mathbb{C}$ except from the jump curve $\gamma_{U}$.

Then the local matrix parametrix $Y^{b}(x;z)$ can be expressed as:
\bee
Y^{b}(x;z)=Y^{b1}(x;z)U_{Bes}(\zeta_b)\left[\!\!\begin{array}{cc}
0& 1 \vspace{0.05in}\\
1& 0
\end{array}\!\!\right]\left(\sqrt{\pm R(z)}f(z)e^{2\sqrt{\zeta_b}-i\pi/4}\right)^{-\sigma_3},\quad z\in\delta^{b}\cap\mathbb{C_{\pm}},
\ene
where
\bee\no
Y^{b1}(x;z)=\frac{\sqrt{2}}{2}Y^o(x;z)\left(\dfrac{\exp(\frac{i\pi}{4})}{\sqrt{\pm R(z)}f(z)}\right)^{\sigma_3}\left[\!\!\begin{array}{cc}
i& 1 \vspace{0.05in}\\
1& i
\end{array}\!\!\right](2\pi\sqrt{\zeta_b})^{\frac{\sigma_3}{2}}.
\ene

Then, we have
\bee
Y^{b}(x;z)\left(Y^o(x;z)\right)^{-1}=\mathbb{I}+\mathcal{O}(|x|^{-1}),\quad x\to-\infty,z\in\partial\delta^{b}\setminus\gamma_U.
\ene

Similarly, we can construct local matrix functions $Y^{a}(x;z),Y^{-a}(x;z)$ and $Y^{-b}(x;z)$ which satisfy the following properties:
\begin{align}\no
\begin{aligned}
&Y^{a}(x;z)\left(Y^o(x;z)\right)^{-1}=\mathbb{I}+\mathcal{O}(|x|^{-1}),\quad x\to-\infty,z\in\partial\delta^{a}\setminus\gamma_U,\v\\
&Y^{-a}(x;z)\left(Y^o(x;z)\right)^{-1}=\mathbb{I}+\mathcal{O}(|x|^{-1}),\quad x\to-\infty,z\in\partial\delta^{-a}\setminus\gamma_U,\v\\
&Y^{-b}(x;z)\left(Y^o(x;z)\right)^{-1}=\mathbb{I}+\mathcal{O}(|x|^{-1}),\quad x\to-\infty,z\in\partial\delta^{-b}\setminus\gamma_U.
\end{aligned}
\end{align}

Then we construct a matrix function as follows:
\bee
Y_3(x;z)
=\begin{cases}
Y^o(x;z),\quad z\in \mathbb{C}\setminus(\delta^{a}\cup\delta^{b}\cup\delta^{-a}\cup\delta^{-b}),\vspace{0.05in}\\
Y^{\pm a}(x;z),\quad z\in\delta^{\pm a},\vspace{0.05in}\\
Y^{\pm b}(x;z),\quad z\in\delta^{\pm b},
\end{cases}
\ene
which has the following jump condition:
 \bee
Y_{3+}(x;z)=Y_{3-}(x;z)V_8(x;z).
\ene

\subsection{The small norm Riemann-Hilbert problem}

Define the error matrix function $E(x;z)$ as follow:
\bee
E(x;z)=Y_2(x;z)Y_3^{-1}(x;z).
\ene

Then error matrix function $E(x;z)$ satisfies the following Riemann-Hilbert problem.

\begin{prop}\label{RH9} 
Find a $2\times 2$ matrix function $E(x;z)$ that satisfies the following properties:

\begin{itemize}

 \item {} Analyticity: $E(x;z)$ is analytic in $\mathbb{C}\setminus\gamma_3$ where $\gamma_3:=O_1\cup O_2\cup\partial\delta^{a}\cup\partial\delta^{b}\cup\partial\delta^{-a}\cup\partial\delta^{-b}\setminus(\delta^{a}\cup\delta^{b}\cup\delta^{-a}\cup\delta^{-b})$ and takes continuous boundary values on $\gamma_3$.

 \item {} Jump condition: The boundary values on the jump contour are defined as
 \bee
E_{+}(x;z)=E_{-}(x;z)V_E(x;z),\quad z\in \gamma_3,
\ene
where
\bee
V_E(x;z)=Y_{3-}(x;z)V_6(x;z)V_8^{-1}(x;z)Y_{3-}^{-1}(x;z),
\ene
which satisfies the following properties:
\bee
V_E(x;z)
=\begin{cases}
\mathbb{I}+\mathcal{O}(e^{-c|x|}),\quad z\in O_1\cup O_2\setminus(\overline{\delta}^a\cup\overline{\delta}^b\cup\overline{\delta}^{-a}\cup\overline{\delta}^{-b}),\v\\
\mathbb{I}+\mathcal{O}(\dfrac{1}{|x|}),\quad z\in\partial\delta^{a}\cup\partial\delta^{b}\cup\partial\delta^{-a}\cup\partial\delta^{-b},
\end{cases}
\ene
with $c$ being a suitable positive, real parameter.

 \item {} Normalization: $E(x;z)=\mathbb{I}+\mathcal{O}(z^{-1}),\quad z\rightarrow\infty.$

\end{itemize}
\end{prop}

According to Riemann-Hilbert problem \ref{RH9}, we obtain
\bee\label{E-error}
E(x;z)=\mathbb{I}+\frac{E_1(x)}{xz}+\mathcal{O}(z^{-2}).
\ene

\begin{propo}
As $x\to-\infty$, the potential function $q(x,0)$ has the following large-space asymptotics
\bee \label{hirota-3}
q(x,0)=(b-a)\dfrac{\vartheta_3(0;\tau)\vartheta_3(\frac{xm_2+m_1}{2\pi i}+\frac12;\tau)}{\vartheta_3(\frac12;\tau)\vartheta_3(\frac{xm_2+m_1}{2\pi i};\tau)}+\mathcal{O}(|x|^{-1}).
\ene
\end{propo}

\begin{proof}
According to Eqs.~(\ref{Y1-Y}),(\ref{Y2-Y1}) and (\ref{E-error}), we have
\begin{align}\no
\begin{aligned}
Y(x;z)&=Y_1(x;z)f(z)^{-\sigma_3}e^{-xg(z)\sigma_3}\v\\
&=Y_2(x;z)f(z)^{-\sigma_3}e^{-xg(z)\sigma_3}\v\\
&=E(x;z)Y_3(x;z)f(z)^{-\sigma_3}e^{-xg(z)\sigma_3}\v\\
&=\left(\mathbb{I}+\frac{E_1(x)}{xz}+\mathcal{O}(z^{-2})\right)Y_3(x;z)f(z)^{-\sigma_3}e^{-xg(z)\sigma_3},
\end{aligned}
\end{align}
which leads to
\begin{align}
\begin{aligned}
Y_{12}(x;z)&=\left(Y_{3,12}(x;z)+\frac{E_{1,12}(x)}{xz}+\mathcal{O}(z^{-2})\right)f(z)e^{xg(z)}\v\\
&=\left(Y^{o}_{12}(x;z)+\frac{E_{1,12}(x)}{xz}+\mathcal{O}(z^{-2})\right)f(z)e^{xg(z)}\v\\
&=\frac12\left(\psi_{12}(x;z)+\frac{E_{1,12}(x)}{xz}+\mathcal{O}(z^{-2})\right)f(z)e^{xg(z)}.
\end{aligned}
\end{align}

Note that
\begin{align}
\begin{aligned}
&f(z)=1+\dfrac{f_1}{z}+\mathcal{O}(z^{-2}),\v\\
&e^{xg(z)}=1+\dfrac{x(\frac{a^2+b^2}{2}+b^2(\frac{K_2(m)}{K_1(m)}-1))}{z}+\mathcal{O}(z^{-2}),\v\\
&m_5(z)=1+\dfrac{a-b}{2z}+\mathcal{O}(z^{-2}).
\end{aligned}
\end{align}

According to Eq.~(\ref{fanyan-5}), one can show the large-space asymptotics (\ref{hirota-3}) of the solution of the Hirota equation.
\end{proof}


Note that similarly, we also consider the asymptotic behavior of $q(x,t)$ as $x\to-\infty$ for each fixed $t=t_0$, where
the norming constant $c_j$ of the residue conditions (\ref{Res1}) of matrix function $M(x,t_0;z)$ is changed as
\bee
c_j=\dfrac{i(b-a)r(iz_j)e^{4i\alpha z_j^2t_0-8\beta z_j^3t_0}}{N\pi},
\ene
where $r(z)$ is a real-valued, continuous and non-vanishing function for $z\in(ia,ib)$, with the symmetry $r(z^*)=r(z)$.
Similarly, as $x\to-\infty$, we have the same large-space asymptotic behavior of the potential function $q(x,t_0)$ as Eq.~(\ref{hirota-3}).

\section{Long-time asymptotics of the potential $q(x,t)$ as $t\to+\infty$}

We recall the Riemann-Hilbert problem \ref{RH4} of function $M^{(3)}(x,t;z)$:
\begin{align}\label{M3-2}
\begin{aligned}
&M_+^{(3)}(x,t;z)=\begin{cases}
M_-^{(3)}(x,t;z)\left[\!\!\begin{array}{cc}
1& 0  \vspace{0.05in}\\
iR(z)e^{-2zx-4i\alpha z^2t+8\beta z^3t}& 1
\end{array}\!\!\right],\quad z\in \gamma_1,\vspace{0.05in}\\
M_-^{(3)}(x,t;z)\left[\!\!\begin{array}{cc}
1& iR(z)e^{2zx+4i\alpha z^2t-8\beta z^3t}  \vspace{0.05in}\\
0& 1
\end{array}\!\!\right],\quad z\in \gamma_2,
\end{cases}\v\\
&M^{(3)}(x,t;z)=\mathbb{I}+\mathcal{O}(z^{-1}),\quad z\rightarrow\infty.
\end{aligned}
\end{align}

Girotti {\it et al}~\cite{Girotti-2} studied the long-time asymptotics of the dense soliton gas of modified KdV equation. Similarly, in what follows we will study the long-time asymptotics of the potential $q(x,t)$ as $t\to +\infty$ based on the Riemann-Hilbert problem. In fact, the following similar results can also be found for $t\to -\infty$.

\subsection{The case $\xi=x/t>4\beta b^2$}

Let $\xi:=\frac{x}{t}$, then Eq.~(\ref{M3-2}) can be rewrite as:
\begin{align}\label{M3-3}
\begin{aligned}
&M_+^{(3)}(x,t;z)=\begin{cases}
M_-^{(3)}(x,t;z)\left[\!\!\begin{array}{cc}
1& 0  \vspace{0.05in}\\
iR(z)e^{-2zt(\xi-4\beta z^2+2i\alpha z)}& 1
\end{array}\!\!\right],\quad z\in \gamma_1,\vspace{0.05in}\\
M_-^{(3)}(x,t;z)\left[\!\!\begin{array}{cc}
1& iR(z)e^{2zt(\xi-4\beta z^2+2i\alpha z)}  \vspace{0.05in}\\
0& 1
\end{array}\!\!\right],\quad z\in \gamma_2
\end{cases}\v\\
&M^{(3)}(x,t;z)=\mathbb{I}+\mathcal{O}(z^{-1}),\quad z\rightarrow\infty.
\end{aligned}
\end{align}

Then we have
\bee
M^{(3)}(x,t;z)=\mathbb{I}+\mathcal{O}(e^{-2ta(\xi-4\beta b^2)}),\quad \xi>4\beta b^2,~t\to+\infty.
\ene
and the potential function $q(x,t)$ becomes trivial.

\subsection{The case $\xi_c<\xi=x/t<4\beta b^2$}

The jump matrices in Eq.~(\ref{M3-3}) contain the two terms $iR(z)e^{\pm 2zt(\xi-4\beta z^2+2i\alpha z)}$. To deal with them conveniently, we only consider them to be of the pure imaginary functions, that is, $\alpha=0,\, \beta\not=0$, in which the Hirota equation reduces to the complex modified KdV equation
\bee\label{cmkdv}
q_t+\beta (q_{xxx}+ 6|q|^2q_x)=0,
\ene
in which one can take $\beta=1$ without loss of generality.

Eq.~(\ref{M3-3}) with $\alpha=0,\,\beta=1$ can be rewritten as:
\begin{align}\label{M3-3g}
\begin{aligned}
&M_+^{(3)}(x,t;z)=\begin{cases}
M_-^{(3)}(x,t;z)\left[\!\!\begin{array}{cc}
1& 0  \vspace{0.05in}\\
iR(z)e^{-2tz(\xi-4z^2)}& 1
\end{array}\!\!\right],\quad z\in \gamma_1,\vspace{0.05in}\\
M_-^{(3)}(x,t;z)\left[\!\!\begin{array}{cc}
1& iR(z)e^{2tz(\xi-4 z^2)}  \vspace{0.05in}\\
0& 1
\end{array}\!\!\right],\quad z\in \gamma_2,
\end{cases} \v\\
&M^{(3)}(x,t;z)=\mathbb{I}+\mathcal{O}(z^{-1}),\quad z\rightarrow\infty.
\end{aligned}
\end{align}

In what follows we consider the condition $\xi_c<\xi<b^2$, where $\xi_c$ to be determined below (see Eq.~(\ref{xi-c})).

Let
\bee\no
\gamma_{1a_1}:=(a_1,b),\quad \gamma_{2a_1}:=(-b,-a_1),\quad a_1\in(a,b).
\ene
Therefore, we construct the following Riemann-Hilbert problem based on the RHP (\ref{M3-3g}).
\begin{prop}\label{RH10} 
Find a $2\times 2$ matrix function $Y_{a_1}(x,t;z)$ that satisfies the following properties:

\begin{itemize}

 \item {} Analyticity: $Y_{a_1}(x,t;z)$ is analytic in $\mathbb{C}\setminus(\gamma_1\cup \gamma_2)$ and takes continuous boundary values on $\gamma_1\cup \gamma_2$.

 \item {} Jump condition: The boundary values on the jump contour $\gamma_1\cup\gamma_2$ are defined as
 \bee
Y_{a_1+}(x,t;z)=Y_{a_1-}(x,t;z)V_{4a_1}(x,t;z),\quad z\in \gamma_1\cup \gamma_2,
\ene
where
\bee
V_{4a_1}(x,t;z)
=\begin{cases}
\left[\!\!\begin{array}{cc}
1& 0  \vspace{0.05in}\\
iR(z)e^{-2tz(\xi-4z^2)}& 1
\end{array}\!\!\right],\quad z\in \gamma_1,\vspace{0.05in}\\
\left[\!\!\begin{array}{cc}
1& iR(z)e^{2tz(\xi-4 z^2)}  \vspace{0.05in}\\
0& 1
\end{array}\!\!\right],\quad z\in \gamma_2;
\end{cases}
\ene

 \item {} Normalization: $Y_{a_1}(x,t;z)=\mathbb{I}+\mathcal{O}(z^{-1}),\quad z\rightarrow\infty.$

\end{itemize}
\end{prop}

According to Eq.~(\ref{fanyan-4}), we can recover $q(x,t)$ by the following formula:
\bee\label{fanyan-6}
q(x,t)=-2\lim\limits_{z\rightarrow\infty}(z Y_{a_1}(x,t;z))_{12}.
\ene

To solve Riemann-Hilbert problem \ref{RH10}, we make the following transformation:
\bee\label{Y1-Y-a1}
Y_{1a_1}(x,t;z)=Y_{a_1}(x,t;z)[f_{a_1}(z)e^{tg_{a_1}(z)}]^{\sigma_3},
\ene
where $f_{a_1}(z)$ and $g_{a_1}(z)$ are scalar functions which satisfy the following properties:
\begin{itemize}

 \item {} Analyticity: $f_{a_1}(z),\, g_{a_1}(z)$ are analytic in $\mathbb{C}\setminus(-b,b)$ and take continuous boundary values on $(-b,b)$.

 \item {} Jump conditions: The boundary values on the jump contour are defined as
\begin{align}
\begin{aligned}
&f_{a_1+}(z)f_{a_1-}(z)=R^{-1}(z),\quad z\in \gamma_{1a_1},\v\\
&f_{a_1+}(z)f_{a_1-}(z)=R(z),\quad z\in \gamma_{2a_1},\v\\
&f_{a_1+}(z)=f_{a_1-}(z)e^{m_{1a_1}},\quad z\in[-a_1,a_1]
\end{aligned}
\end{align}
and
\begin{align}
\begin{aligned}
&g_{a_1+}(z)+g_{a_1-}(z)=2\xi z-8z^3,\quad z\in \gamma_{1a_1}\cup\gamma_{2a_1},\v\\
&g_{a_1+}(z)-g_{a_1-}(z)=m_{2a_1},\quad z\in[-a_1,a_1] ,
\end{aligned}
\end{align}
where
\begin{align}\no
m_{1a_1}=\dfrac{\displaystyle\int_{a_1}^{b}\dfrac{2\ln R(\lambda)}{C_{a_1+}(\lambda)}d\lambda}{\displaystyle\int_{-a_1}^{a_1}\dfrac{d\lambda}{C_{a_1}(\lambda)}},\quad
C_{a_1}(z)=\sqrt{(z^2-a_1^2)(z^2-b^2)},
\end{align}
\bee\no
m_{2a_1}=-\dfrac{i\pi b(2a_1^2+2b^2-\xi)}{K_1(m_{a_1})},\quad m_{a_1}=\frac{a_1}{b}.
\ene

 \item {} Normalization: $f_{a_1}(z)=1+\mathcal{O}(z^{-1}),\quad g_{a_1}(z)=\mathcal{O}(z^{-1}), \quad z\rightarrow\infty$.
\end{itemize}

These scalar functions $f_{a_1}(z)$ and $g_{a_1}(z)$ can also be accurately solved as~\cite{Girotti-1,Girotti-2}:
\begin{align}
\begin{aligned}
&f_{a_1}(z)=\exp\left\{\frac{C_{a_1}(z)}{2\pi i}\left(\int_{\gamma_{2a_1}-\gamma_{1a_1}}\frac{\ln R(\lambda)}{C_{a_1+}(\lambda)(\lambda-z)}d\lambda
+\int_{-a_1}^{a_1}\frac{m_{1a_1}}{C_{a_1}(\lambda)(\lambda-z)}d\lambda
\right)\right\},\v\v\\
&g_{a_1}(z)=-4z^3+\xi z+\int_{b}^{z}\frac{12z^4-z^2[\xi+6(a_1^2+b^2)]+4a_1^2b^2
+b^2[2(a_1^2+b^2)-\xi]\left(\frac{K_2(m_{a_1})}{K_1(m_{a_1})}-1\right)}{C_{a_1}(\lambda)}d\lambda,
\end{aligned}
\end{align}
where the parameter $a_1$ is determined by the following equation:
\bee 
\xi=2(a_1^2+b^2)+\dfrac{4a_1^2(a_1^2-b^2)}{a_1^2+b^2\left(\frac{K_2(m_{a_1})}{K_1(m_{a_1})}-1\right)},
\ene
such that one has
\bee\label{xi-c}
\xi_c=2(a^2+b^2)+\dfrac{4a^2(a^2-b^2)}{a^2+b^2\left(\frac{K_2(m)}{K_1(m)}-1\right)},\quad m=\frac{a}{b}.
\ene

Therefore, the matrix function $Y_{1a_1}(x,t;z)$ satisfies the following Riemann-Hilbert problem.

\begin{prop}\label{RH11} 
Find a $2\times 2$ matrix function $Y_{1a_1}(x,t;z)$ that satisfies the following properties:

\begin{itemize}

 \item {} Analyticity: $Y_{1a_1}(x,t;z)$ is analytic in $\mathbb{C}\setminus(-b,b)$ and takes continuous boundary values on $(-b,b)$.

 \item {} Jump condition: The boundary values on the jump contour are defined as
 \bee
Y_{1a_1+}(x,t;z)=Y_{1a_1-}(x,t;z)V_{5a_1}(x,t;z),\quad z\in (-b,b),
\ene
where
\bee
V_{5a_1}(x,t;z)
=\begin{cases}
\left[\!\!\begin{array}{cc}
\delta f_{a_1}(z)e^{t\Delta g_{a_1}(z)}& 0  \vspace{0.05in}\\
i& \delta^{-1}f_{a_1}(z) e^{-t\Delta g_{a_1}(z)}
\end{array}\!\!\right],\quad z\in \gamma_{1a_1},\vspace{0.05in}\\
\left[\!\!\begin{array}{cc}
e^{tm_{2a_1}+m_{1a_1}}& 0  \vspace{0.05in}\\
iR(z)\widehat\delta f_{a_1}(z)e^{t(\widehat\Delta g_{a_1}(z)-2\xi z+8z^3)}& e^{-(tm_{2a_1}+m_{1a_1})}
\end{array}\!\!\right],\quad z\in [a,a_1],\vspace{0.05in}\\
\left[\!\!\begin{array}{cc}
e^{tm_{2a_1}+m_{1a_1}}& 0  \vspace{0.05in}\\
0& e^{-(tm_{2a_1}+m_{1a_1})}
\end{array}\!\!\right],\quad z\in [-a,a],\vspace{0.05in}\\
\left[\!\!\begin{array}{cc}
e^{tm_{2a_1}+m_{1a_1}}& \dfrac{iR(z)e^{-t(\widehat\Delta g_{a_1}(z)-2\xi z+8z^3)}}{\widehat\delta f_{a_1}(z)}  \vspace{0.05in}\\
0& e^{-(tm_{2a_1}+m_{1a_1})}
\end{array}\!\!\right],\quad z\in [-a_1,-a],\vspace{0.05in}\\
\left[\!\!\begin{array}{cc}
\widehat\delta f_{a_1}(z) e^{t\widehat\Delta g_{a_1}(z)} & i  \vspace{0.05in}\\
0& \dfrac{e^{-t\widehat\Delta g_{a_1}(z)}}{\widehat\delta f_{a_1}(z)}
\end{array}\!\!\right],\quad z\in \gamma_{2a_1};
\end{cases}
\ene
with  $\delta f_{a_1}(z)=f_{a_1+}(z)/f_{a_1-}(z),\, \widehat\delta f_{a_1}(z)=f_{a_1+}(z)f_{a_1-}(z),\, \Delta g_{a_1}(z)=g_{a_1+}(z)-g_{a_1-}(z),\, \widehat\Delta g_{a_1}(z)=g_{a_1+}(z)+g_{a_1-}(z).$

 \item {} Normalization: $Y_{1a_1}(x,t;z)=\mathbb{I}+\mathcal{O}(z^{-1}),\quad z\rightarrow\infty$.
\end{itemize}
\end{prop}

\subsubsection{Opening lenses}

Based on the idea~\cite{Girotti-2}, we open the lens $O_1$ to pass through points $z=a_1$ and $z=b$ and lens $O_2$ to pass through points $z=-b$ and $z=-a_1$ (see Fig.~\ref{fig3}). We define a new function as follows:
\bee
R_{\pm}(z):=\pm R(z),\quad z\in\gamma_{1a_1}\cup\gamma_{2a_1}.
\ene

\begin{figure}[!t]
    \centering
 \vspace{-0.15in}
  {\scalebox{0.76}[0.76]{\includegraphics{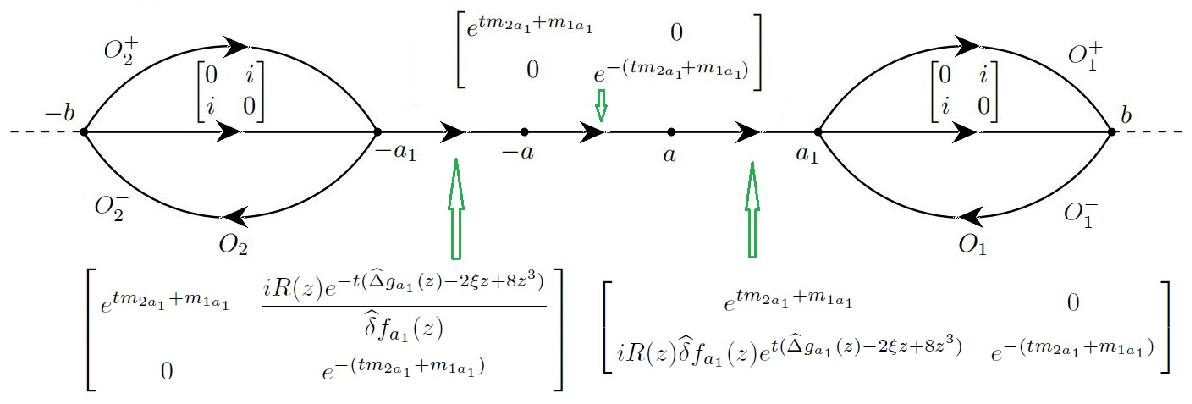}}}\hspace{-0.35in}
\vspace{0.15in}
\caption{Riemann-Hilbert problem \ref{RH12} of function $Y_{2a_1}(x,t;z)$. Opening lenses $O_1$ and $O_2$: $t\to+\infty$.}
   \label{fig3}
\end{figure}

Note that the jump matrix $V_5(x,t;z)|_{\gamma_{1a_1}}$ has the following decomposition:
\begin{align}
\begin{aligned}
V_{5a_1}(x,t;z)&=\left[\!\!\begin{array}{cc}
\delta f_{a_1}(z)e^{t\Delta g_{a_1}(z)}& 0  \vspace{0.05in}\\
i& \delta^{-1}f_{a_1}(z) e^{-t\Delta g_{a_1}(z)}
\end{array}\!\!\right]\vspace{0.05in}\\
&=\left[\!\!\begin{array}{cc}
1& \dfrac{ie^{-2t(g_{a_1-}(z)+4z^3-\xi z)}}{R_{-}(z)f_{a_1-}^2(z)}  \vspace{0.05in}\\
0& 1
\end{array}\!\!\right]\left[\!\!\begin{array}{cc}
0& i  \vspace{0.05in}\\
i& 0
\end{array}\!\!\right]\left[\!\!\begin{array}{cc}
1& -\dfrac{ie^{-2t(g_{a_1+}(z)+4z^3-\xi z)}}{R_{+}(z)f_{a_1+}^2(z)}  \vspace{0.05in}\\
0& 1
\end{array}\!\!\right],\quad z\in \gamma_{1a_1},
\end{aligned}
\end{align}
and the jump matrix $V_5(x,t;z)|_{\gamma_{2a_1}}$ has the following decomposition:
\begin{align}
\begin{aligned}
V_{5a_1}(x,t;z)&=\left[\!\!\begin{array}{cc}
\widehat\delta f_{a_1}(z) e^{t\widehat\Delta g_{a_1}(z)} & i  \vspace{0.05in}\\
0& \dfrac{e^{-t\widehat\Delta g_{a_1}(z)}}{\widehat\delta f_{a_1}(z)}
\end{array}\!\!\right]\vspace{0.05in}\\
&=\left[\!\!\begin{array}{cc}
1& 0  \vspace{0.05in}\\
\dfrac{ie^{2t(g_{a_1-}(z)+4z^3-\xi z)}f_-^2(z)}{R_{-}(z)}& 1
\end{array}\!\!\right]\left[\!\!\begin{array}{cc}
0& i  \vspace{0.05in}\\
i& 0
\end{array}\!\!\right]\left[\!\!\begin{array}{cc}
1& 0  \vspace{0.05in}\\
-\dfrac{ie^{2t(g_{a_1+}(z)+4z^3-\xi z)}f_+^2(z)}{R_{+}(z)}& 1
\end{array}\!\!\right],\quad z\in \gamma_{2a_1}.
\end{aligned}
\end{align}

We can proceed with opening lenses. Make the following transformation:
\bee\label{Y2-Y1-a1}
Y_{2a_1}(x,t;z)=\begin{cases}
Y_{1a_1}(x,t;z)\left[\!\!\begin{array}{cc}
1& \dfrac{ie^{-2t(g_{a_1}(z)+4z^3-\xi z)}}{R(z)f^2(z)}  \vspace{0.05in}\\
0& 1
\end{array}\!\!\right],\quad \mathrm{in~the~upper/lower~lens}~O_1,\vspace{0.05in}\\
Y_{1a_1}(x,t;z)\left[\!\!\begin{array}{cc}
1& 0  \vspace{0.05in}\\
\dfrac{ie^{2t(g_{a_1}(z)+4z^3-\xi z)}f^2(z)}{R(z)}& 1
\end{array}\!\!\right],\quad \mathrm{in~the~upper/lower~lens}~O_2,\vspace{0.05in}\\
Y_{1a_1}(x,t;z),\quad \mathrm{otherwise}.
\end{cases},
\ene

Then matrix function $Y_{2a_1}(x,t;z)$ satisfies the following Riemann-Hilbert problem.

\begin{prop}\label{RH12} 
Find a $2\times 2$ matrix function $Y_{2a_1}(x,t;z)$ that satisfies the following properties:

\begin{itemize}

 \item {} Analyticity: $Y_2(x,t;z)$ is analytic in $\mathbb{C}\setminus((-b,b)\cup O_1\cup O_2)$ and takes continuous boundary values on $(-b,b)\cup O_1\cup O_2$.

 \item {} Jump condition: The boundary values on the jump contour are defined as
 \bee
Y_{2a_1+}(x,t;z)=Y_{2a_1-}(x,t;z)V_{6a_1}(x,t;z),\quad z\in (-b,b)\cup O_1^+\cup O_1^-\cup O_2^+\cup O_2^-,
\ene
where for $z\in O_1^+\cup O_1^-\cup O_2^+\cup O_2^-$
\bee\label{V6-a1}
V_{6a_1}(x;z)
=\begin{cases}
\left[\!\!\begin{array}{cc}
1& -\dfrac{ie^{-2t(g_{a_1}(z)+4z^3-\xi z)}}{R(z)f^2(z)}  \vspace{0.05in}\\
0& 1
\end{array}\!\!\right],\quad z\in O_1^+\cup O_1^-,\vspace{0.05in}\\
\left[\!\!\begin{array}{cc}
1& 0  \vspace{0.05in}\\
-\dfrac{ie^{2t(g_{a_1}(z)+4z^3-\xi z)}f^2(z)}{R(z)}& 1
\end{array}\!\!\right],\quad z\in O_2^+\cup O_2^-,\v\\
\end{cases}
\ene
and for $z\in (-b,b)$
\bee
V_{6a_1}(x;z)
=\begin{cases}
\left[\!\!\begin{array}{cc}
0& i  \vspace{0.05in}\\
i& 0
\end{array}\!\!\right],\quad z\in \gamma_{1a_1}\cup \gamma_{2a_1},\vspace{0.05in}\\
\left[\!\!\begin{array}{cc}
e^{tm_{2a_1}+m_{1a_1}}& 0  \vspace{0.05in}\\
iR(z)\widehat\delta f_{a_1}(z)e^{t(\widehat\Delta g_{a_1}(z)-2\xi z+8z^3)}& e^{-(tm_{2a_1}+m_{1a_1})}
\end{array}\!\!\right],\quad z\in [a,a_1],\vspace{0.05in}\\
\left[\!\!\begin{array}{cc}
e^{tm_{2a_1}+m_{1a_1}}& 0  \vspace{0.05in}\\
0& e^{-(tm_{2a_1}+m_{1a_1})}
\end{array}\!\!\right],\quad z\in [-a,a],\vspace{0.05in}\\
\left[\!\!\begin{array}{cc}
e^{tm_{2a_1}+m_{1a_1}}& \dfrac{iR(z)e^{-t(\widehat\Delta g_{a_1}(z)-2\xi z+8z^3)}}{\widehat\delta f_{a_1}(z)}  \vspace{0.05in}\\
0& e^{-(tm_{2a_1}+m_{1a_1})}
\end{array}\!\!\right],\quad z\in [-a_1,-a],\vspace{0.05in}\\
\end{cases}
\ene

 \item {} Normalization: $Y_{2a_1}(x,t;z)=\mathbb{I}+\mathcal{O}(z^{-1}),\quad z\rightarrow\infty.$

\end{itemize}
\end{prop}

\begin{lemma}~\cite{Girotti-1}\label{le-a1} The following inequalities hold:
\begin{align}
\begin{aligned}
&\mathrm{Re}(g_{a_1}(z)+4z^3-\xi z)>0,\quad z\in O_1\setminus\{a_1,b\},\v\\
&\mathrm{Re}(g_{a_1}(z)+4z^3-\xi z)<0,\quad z\in O_2\setminus\{-a_1,-b\},\v\\
&\mathrm{Re}(g_{a_1+}(z)+g_{a_1-}(z)-2\xi z+8z^3)<0,\quad z\in [a,a_1),\v\\
&\mathrm{Re}(g_{a_1+}(z)+g_{a_1-}(z)-2\xi z+8z^3)>0,\quad z\in (-a_1,-a].
\end{aligned}
\end{align}

\end{lemma}

According to Lemma \ref{le-a1}, we know that the off-diagonal entries of the jump matrix defined by Eq.~(\ref{V6-a1}) along the
upper and lower lenses $O_1,O_2$ are exponentially decay when $t\to+\infty$.

\subsubsection{The outer parametrix}

In order to solve the Riemann-Hilbert problem \ref{RH12}, we will introduce a two-sheeted Riemann surface $\Omega$ associated to the function
function $C_{a_1}(z)$ (see Fig.~\ref{fig4})
\bee
\Omega:=\left\{(\zeta_1,\zeta_2)\in\mathbb{C}^2|\zeta_2^2=C_{a_1}^2(z)=(z^2-a_1^2)(z^2-b^2)\right\}.
\ene

\begin{figure}[!t]
    \centering
 \vspace{-0.15in}
  {\scalebox{0.55}[0.55]{\includegraphics{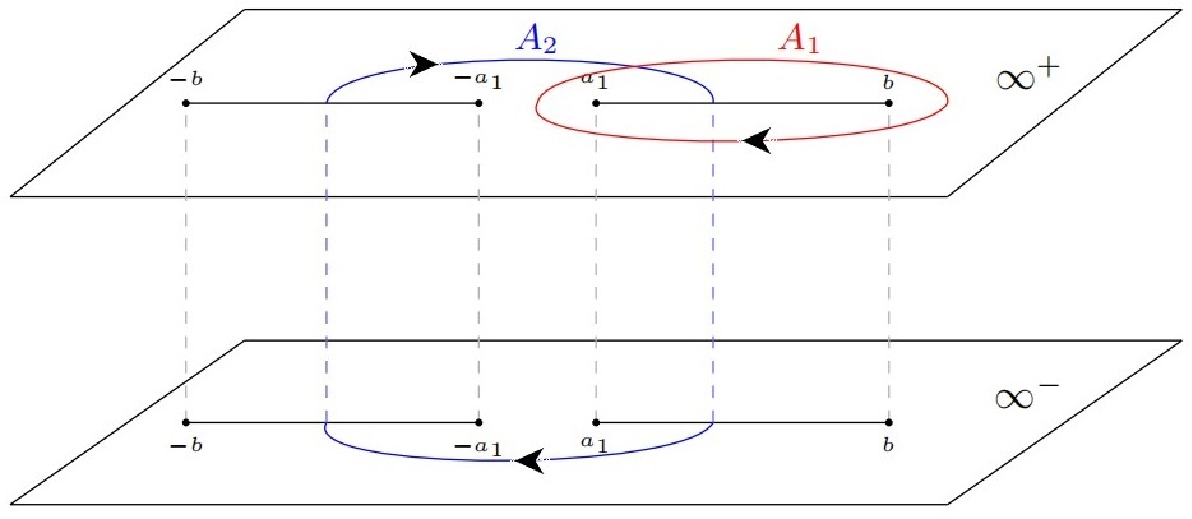}}}\hspace{-0.35in}
\vspace{0.05in}
\caption{Riemann surface $\Omega$.}
   \label{fig4}
\end{figure}
We construct a new Riemann-Hilbert problem which only has jumps on $(-b,b)$.

\begin{prop}\label{RH13} 
Find a $2\times 2$ matrix function $Y^{o}(x,t;z)$ that satisfies the following properties:

\begin{itemize}

 \item {} Analyticity: $Y^o(x,t;z)$ is analytic in $\mathbb{C}\setminus(-b,b)$ and takes continuous boundary values on $(-b,b)$.

 \item {} Jump condition: The boundary values on the jump contour are defined as
 \bee
Y^o_{+}(x,t;z)=Y^o_{-}(x,t;z)V_{7a_1}(x,t;z),\quad z\in \gamma_{1a_1}\cup[-a_1,a_1]\cup\gamma_{2a_1},
\ene
where
\bee
V_{7a_1}(x,t;z)
=\begin{cases}
\left[\!\!\begin{array}{cc}
0& i  \vspace{0.05in}\\
i& 0
\end{array}\!\!\right],\quad z\in \gamma_{1a_1}\cup \gamma_{2a_1},\vspace{0.05in}\\
\left[\!\!\begin{array}{cc}
e^{tm_{2a_1}+m_{1a_1}}& 0  \vspace{0.05in}\\
0& e^{-(tm_{2a_1}+m_{1a_1})}
\end{array}\!\!\right],\quad z\in [-a_1,a_1].
\end{cases}
\ene

 \item {} Normalization: $Y^o(x,t;z)=\mathbb{I}+\mathcal{O}(z^{-1}),\quad z\rightarrow\infty.$
\end{itemize}
\end{prop}

Then, we define the functions:
\begin{align}
\begin{aligned}
&\tau:=\oint_{A_1}-\dfrac{b}{4K_1(m_{a_1})C_{a_1}(\lambda)}d\lambda,\v\\
&m_{3a_1}(z)=\int_{b}^{z}-\dfrac{b}{4K_1(m_{a_1})C_{a_1}(\lambda)}d\lambda,\v\\ &m_{5a_1}(z)=\left(\frac{(z-b)(z+a_1)}{(z-a_1)(z+b)}\right)^{\frac14},
\end{aligned}
\end{align}
which have the following jump relations:
\begin{align}
\begin{aligned}
&m_{3a_1+}(z)+m_{3a_1-}(z)=0,\quad z\in\gamma_{1a_1},\v\\
&m_{3a_1+}(z)-m_{3a_1-}(z)=-\tau,\quad z\in(-a_1,a_1),\v\\
&m_{3a_1+}(z)+m_{3a_1-}(z)=-1,\quad z\in\gamma_{2a_1},\v\\
&m_{3a_1+}(z)-m_{3a_1-}(z)=0,\quad z\in[b,+\infty),\v\\
&m_{5a_1+}(z)=im_{5a_1-}(z),\quad z\in\gamma_{1a_1}\cup\gamma_{2a_1},
\end{aligned}
\end{align}

\begin{lemma}~\cite{Girotti-1}\label{le4-a1} Define the functions:
\begin{align}
\begin{aligned}
&\psi_{11}(z)=(m_{5a_1}(z)+m_{5a_1}^{-1}(z))\dfrac{\vartheta_3(0;\tau)\vartheta_3(m_{3a_1}(z)+\frac{tm_{2a_1}+m_{1a_1}}{2\pi i}+\frac14;\tau)}{\vartheta_3(m_{3a_1}(z)+\frac14;\tau)\vartheta_3(\frac{tm_{2a_1}+m_{1a_1}}{2\pi i};\tau)},\v\\
&\psi_{21}(z)=(m_{5a_1}(z)-m_{5a_1}^{-1}(z))\dfrac{\vartheta_3(0;\tau)\vartheta_3(m_{3a_1}(z)+\frac{tm_{2a_1}+m_{1a_1}}{2\pi i}-\frac14;\tau)}{\vartheta_3(m_{3a_1}(z)-\frac14;\tau)\vartheta_3(\frac{tm_{2a_1}+m_{1a_1}}{2\pi i};\tau)},\v\\
&\psi_{12}(z)=(m_{5a_1}(z)-m_{5a_1}^{-1}(z))\dfrac{\vartheta_3(0;\tau)\vartheta_3(-m_{3a_1}(z)+\frac{tm_{2a_1}+m_{1a_1}}{2\pi i}+\frac14;\tau)}{\vartheta_3(-m_{3a_1}(z)+\frac14;\tau)\vartheta_3(\frac{tm_{2a_1}+m_{1a_1}}{2\pi i};\tau)},\v\\
&\psi_{22}(z)=(m_{5a_1}(z)+m_{5a_1}^{-1}(z))\dfrac{\vartheta_3(0;\tau)\vartheta_3(-m_{3a_1}(z)+\frac{tm_{2a_1}+m_{1a_1}}{2\pi i}-\frac14;\tau)}{\vartheta_3(-m_{3a_1}(z)-\frac14;\tau)\vartheta_3(\frac{tm_{2a_1}+m_{1a_1}}{2\pi i};\tau)},
\end{aligned}
\end{align}
which have the following jump relations:
\begin{align}
\begin{aligned}
&\psi_{11+}(z)=i\psi_{12-}(z),\quad \psi_{21+}(z)=i\psi_{22-}(z),\quad z\in\gamma_{1a_1}\cup\gamma_{2a_1},\v\\
&\psi_{12+}(z)=i\psi_{11-}(z),\quad \psi_{22+}(z)=i\psi_{21-}(z),\quad z\in\gamma_{1a_1}\cup\gamma_{2a_1},\v\\
&\psi_{11+}(z)=\psi_{11-}(z)e^{tm_{2a_1}+m_{1a_1}},\quad \psi_{12+}(z)=\psi_{12-}(z)e^{-(tm_{2a_1}+m_{1a_1})},\quad z\in(-a_1,a_1),\v\\
&\psi_{21+}(z)=\psi_{21-}(z)e^{tm_{2a_1}+m_{1a_1}},\quad \psi_{22+}(z)=\psi_{22-}(z)e^{-(tm_{2a_1}+m_{1a_1})},\quad z\in(-a_1,a_1).
\end{aligned}
\end{align}

\end{lemma}
\begin{proof}
According to Lemma \ref{le4}, the above conclusion can be obtained.
\end{proof}

\begin{propo} The solution of Riemann-Hilbert problem \ref{RH13} can be expressed in the following form.
\bee\label{Yoxz}
Y^o(x,t;z)=\frac12\left[\!\!\begin{array}{cc}
\psi_{11}(x,t;z)& \psi_{12}(x,t;z)  \vspace{0.05in}\\
\psi_{21}(x,t;z)& \psi_{22}(x,t;z)
\end{array}\!\!\right].
\ene

\end{propo}

\begin{proof} According to Lemma \ref{le4-a1}, we know that the matrix function $Y^o(x,t;z)$ is a solution of Riemann-Hilbert problem \ref{RH13}.

\end{proof}

\subsubsection{The local matrix parametrix}

In this section, we will construct a local matrix parametrix $Y^{-a_1}(x,t;z)$ as $z\in\delta^{-a_1}:=\{z||z+a_1|<\epsilon,~\epsilon~\mathrm{is~a~suitable~small~ positive~parameter} \}$. Similarly, we can define regions $\delta^{a_1}$, $\delta^{b}$, and $\delta^{-b}$. Then, let the conformal map be~\cite{Girotti-1}:
\bee
\zeta_{-a_1}:=\left(\frac{3t}{4}\right)^{\frac23}\left(\int_{-a_1}^z[g_{a_1+}'(\lambda)-g_{a_1-}'(\lambda)]d\lambda\right)^{\frac23},\quad z\in\delta^{-a_1}.
\ene
and make the following transformation:
\bee
Y^{(1)}(x,t;z)
=\begin{cases}
Y_{2a_1}(x,t;z)\left(\dfrac{\sqrt{R(z)}}{f_{a_1}(z)}\right)^{\sigma_3}
e^{\left(\frac{i\pi}{4}-\frac12(tm_{2a_1}+m_{1a_1})-\frac23\zeta_{-a_1}^{\frac32}\right)\sigma_3},\quad z\in\mathbb{C}_+\cap\delta^{-a_1},\v\\
Y_{2a_1}(x,t;z)\left(\dfrac{\sqrt{-R(z)}}{f_{a_1}(z)}\right)^{\sigma_3}
e^{\left(\frac{i\pi}{4}+\frac12(tm_{2a_1}+m_{1a_1})-\frac23\zeta_{-a_1}^{\frac32}\right)\sigma_3},\quad z\in\mathbb{C}_-\cap\delta^{-a_1}.
\end{cases}
\ene

\begin{lemma}~\cite{Girotti-1}  Matrix function $Y^{(1)}(x,t;z)$ satisfies the following jump conditions:
\bee
Y_+^{(1)}(x,t;z)
=\begin{cases}
Y_-^{(1)}(x,t;z)\left[\!\!\begin{array}{cc}
1& 0 \vspace{0.05in}\\
1& 1
\end{array}\!\!\right],\quad z\in\delta^{-a_1}\cap\{\mathrm{upper~and~lower~lenses}\},\v\\
Y_-^{(1)}(x,t;z)\left[\!\!\begin{array}{cc}
1& 1 \vspace{0.05in}\\
0& 1
\end{array}\!\!\right],\quad z\in\delta^{-a_1}\cap(0,+\infty),\v\\
Y_-^{(1)}(x,t;z)\left[\!\!\begin{array}{cc}
0& 1 \vspace{0.05in}\\
-1& 0
\end{array}\!\!\right],\quad z\in\delta^{-a_1}\cap(-\infty,0).
\end{cases}
\ene

\end{lemma}

\begin{figure}[!t]
    \centering
 \vspace{-0.15in}
  {\scalebox{0.38}[0.38]{\includegraphics{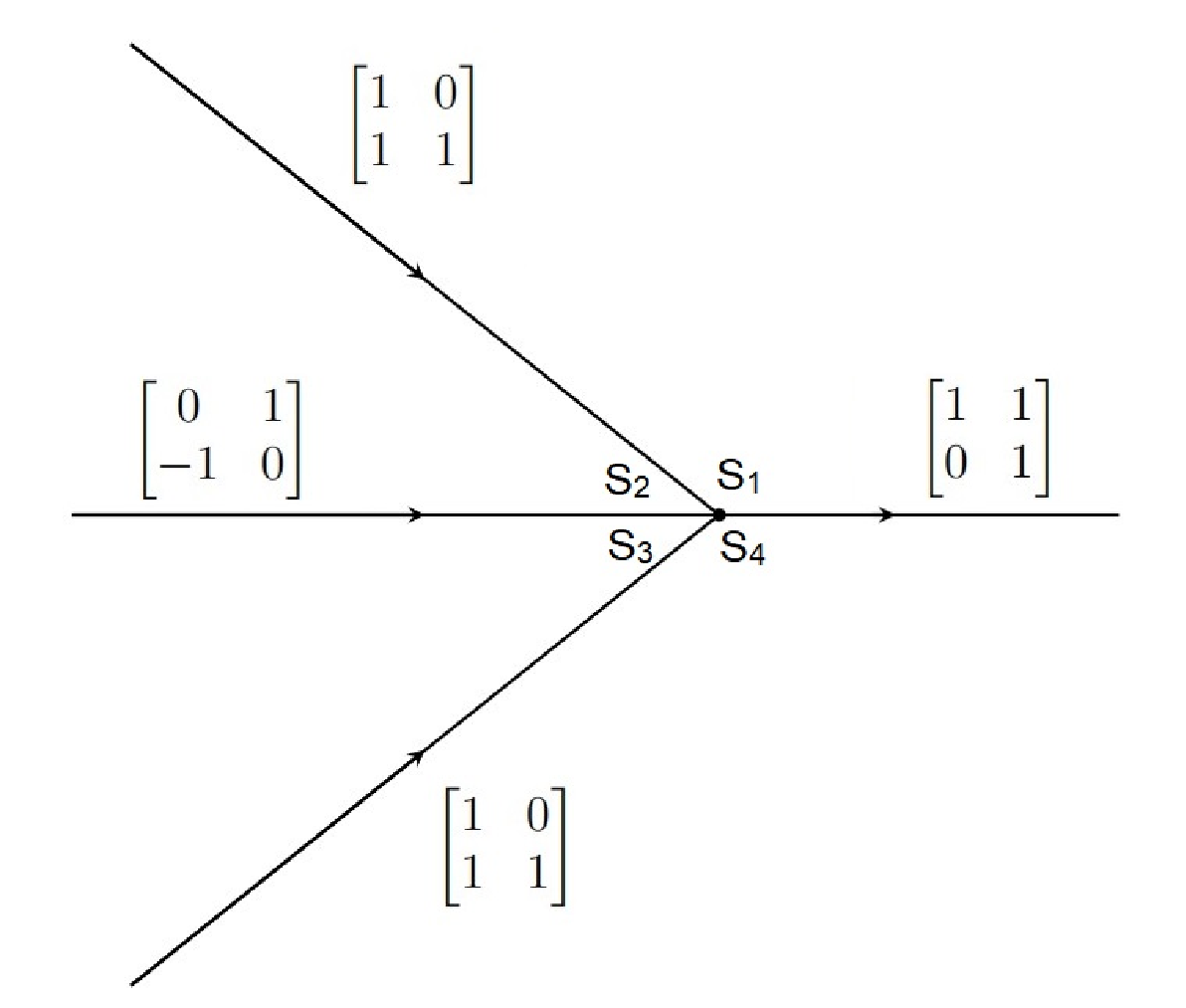}}}\hspace{-0.35in}
\vspace{0.05in}
\caption{Regional division for the Airy model $U_{Ai}(z)$.}
   \label{fig6}
\end{figure}

Next, we will introduce the Airy model $U_{Ai}(z)$~\cite{Deift-1} which satisfies the following Riemann-Hilbert problem:
\begin{prop}\label{Airy}
Find a $2\times 2$ matrix function $U_{Ai}(z)$ that satisfies the following properties:

\begin{itemize}

 \item {} Analyticity: $U_{Ai}(z)$ is analytic for $z$ in the three regions shown in Figure \ref{fig6}, namely, $S_1: 0<\mathrm{arg}(z)<\frac{2\pi}{3}$, $S_2: \frac{2\pi}{3}<\mathrm{arg}(z)<\pi$, $S_3: -\pi<\mathrm{arg}(z)<-\frac{2\pi}{3}$ and $S_4: -\frac{2\pi}{3}<\mathrm{arg}(z)<0$ where $-\pi<\mathrm{arg}(z)\leq\pi$. It takes continuous boundary values on the excluded rays and at the origin from each sector.

 \item {} Jump condition: The boundary values on the jump contour $\gamma_A$ where $\gamma_A:=\gamma_{\pm}\cup\gamma_{0\pm}$ with $\gamma_{\pm}=\{
 \mathrm{arg}(z)=\pm\frac{2\pi}{3}\},\gamma_{0+}=\{
 \mathrm{arg}(z)=0\}$ and $\gamma_{0-}=\{
 \mathrm{arg}(z)=\pi\}$ are defined as:
 \bee
U_{Ai+}(z)=U_{Ai-}(z)V_{Ai}(z),\quad z\in \gamma_A,
\ene
where
\bee
V_{Ai}(z)
=\begin{cases}
\left[\!\!\begin{array}{cc}
1& 0 \\
1& 1
\end{array}\!\!\right],\quad z\in \gamma_{\pm},\v\\
\left[\!\!\begin{array}{cc}
1& 1 \\
0& 1
\end{array}\!\!\right],\quad z\in \gamma_{0+},\v\\
\left[\!\!\begin{array}{cc}
0& 1 \\
-1& 0
\end{array}\!\!\right],\quad z\in \gamma_{0-};
\end{cases}
\ene

 \item {} Normalization: $U_{Ai}(z)$ is bounded as $z\to 0,z\in\mathbb{C}\setminus\gamma_A$.

\end{itemize}

\end{prop}

The solution of the Riemann-Hilbert Problem \ref{Airy} can be expressed explicitly according to the Airy function $Ai(z)$.  In particular, the solution satisfies:
\bee
U_{Ai}(z)=\frac{\sqrt{2}}{2}z^{-\frac14\sigma_3}\left[\!\!\begin{array}{cc}
1& i \vspace{0.05in}\\
i& 1
\end{array}\!\!\right]\left(\mathbb{I}+\mathcal{O}(z^{-\frac32})\right)e^{-\frac23z^{\frac32}\sigma_3},\quad z\to \infty,
\ene
in $\mathbb{C}$ except for the jumps.

Then the local matrix parametrix $Y^{-a_1}(x,t;z)$ can be expressed as:
\begin{align}\no
\begin{aligned}
Y^{-a_1}(x,t;z)=Y^{-a_{11}}(x,t;z)U_{Ai}(\zeta_{-a_1})\left(\dfrac{\sqrt{\pm R(z)}}{f_{a_1}(z)}\right)^{-\sigma_3}e^{-\left(\frac{\pi i}{4}\mp\frac12(tm_{2a_1}+m_{1a_1})-\frac23\zeta_{-a_1}^{\frac32}\right)\sigma_3}&,\v\\
\quad z\in\delta^{-a_1}\cap\mathbb{C_{\pm}}&,
\end{aligned}
\end{align}
where
\bee\no
Y^{-a_{11}}(x,t;z)=\frac{\sqrt{2}}{2}Y^o(x,t;z)e^{\left(\frac{i\pi }{4}\mp\frac12(tm_{2a_1}+m_{1a_1})\right)\sigma_3}\left(\dfrac{\sqrt{\pm R(z)}}{f_{a_1}(z)}\right)^{\sigma_3}\left[\!\!\begin{array}{cc}
1& -i \vspace{0.05in}\\
-i& 1
\end{array}\!\!\right](\zeta_{-a_1})^{\frac14\sigma_3}.
\ene

Then, we have
\bee
Y^{-a_1}(x,t;z)\left(Y^o(x,t;z)\right)^{-1}=\mathbb{I}+\mathcal{O}(t^{-1}),\quad t\to+\infty,z\in\partial\delta^{-a_1}\setminus\gamma_A.
\ene

Similarly, we can construct the local matrix function $Y^{a_1}(x,t;z)$ which satisfies the following property:
\bee
Y^{a_1}(x,t;z)\left(Y^o(x,t;z)\right)^{-1}=\mathbb{I}+\mathcal{O}(t^{-1}),\quad t\to+\infty,z\in\partial\delta^{a_1}\setminus\gamma_A.
\ene

Similar to the case of $x\to-\infty$, we can provide the local function $Y^{b}(x,t;z)$ when $z\to b$ and $Y^{-b}(x,t;z)$ when $z\to -b$ which correspond to the Bessel model $U_{Bes}(z)$. And local matrix functions $Y^{b}(x,t;z)$ and $Y^{-b}(x,t;z)$ satisfy the following properties:
\begin{align}
\begin{aligned}
&Y^{b}(x,t;z)\left(Y^o(x,t;z)\right)^{-1}=\mathbb{I}+\mathcal{O}(t^{-1}),\quad t\to+\infty,z\in\partial\delta^{b}\setminus\gamma_A,\v\\
&Y^{-b}(x,t;z)\left(Y^o(x,t;z)\right)^{-1}=\mathbb{I}+\mathcal{O}(t^{-1}),\quad t\to+\infty,z\in\partial\delta^{-b}\setminus\gamma_A.
\end{aligned}
\end{align}

Then we construct a matrix function as follows:
\bee
Y_{3a_1}(x,t;z)
=\begin{cases}
Y^o(x;z),\quad z\in \mathbb{C}\setminus(\delta^{a_1}\cup\delta^{b}\cup\delta^{-a_1}\cup\delta^{-b}),\vspace{0.05in}\\
Y^{\pm a_1}(x,t;z),\quad z\in\delta^{\pm a_1},\vspace{0.05in}\\
Y^{\pm b}(x,t;z),\quad z\in\delta^{\pm b},
\end{cases}
\ene
which has the following jump condition:
 \bee
Y_{3a_1+}(x,t;z)=Y_{3a_1-}(x,t;z)V_{8a_1}(x,t;z).
\ene

\subsubsection{The small norm Riemann-Hilbert problem}

Let the error matrix function $E(x,t;z)$ as follow:
\bee
E(x,t;z)=Y_{2a_1}(x,t;z)Y_{3a_1}^{-1}(x,t;z).
\ene
Then the error matrix function $E(x,t;z)$ satisfies the following Riemann-Hilbert problem.

\begin{prop}\label{RH14} 
Find a $2\times 2$ matrix function $E(x,t;z)$ that satisfies the following properties:

\begin{itemize}

 \item {} Analyticity: $E(x,t;z)$ is analytic in $\mathbb{C}\setminus\gamma_3$ where $\gamma_3:=(a,a_1)\cup(-a_1,-a)\cup O_1\cup O_2\cup\partial\delta^{a_1}\cup\partial\delta^{b}\cup\partial\delta^{-a_1}\cup\partial\delta^{-b}\setminus(\delta^{a_1}\cup\delta^{b}\cup\delta^{-a_1}\cup\delta^{-b})$ and takes continuous boundary values on $\gamma_3$.

 \item {} Jump condition: The boundary values on the jump contour are defined as
 \bee
E_{+}(x,t;z)=E_{-}(x,t;z)V_E(x,t;z),\quad z\in \gamma_3,
\ene
where
\bee
V_E(x,t;z)=Y_{3a_1-}(x,t;z)V_{6a_1}(x,t;z)V_{8a_1}^{-1}(x,t;z)Y_{3a_1-}^{-1}(x,t;z),
\ene
which satisfies the following properties:
\bee
V_E(x,t;z)
=\begin{cases}
\mathbb{I}+\mathcal{O}(e^{-ct}),\quad z\in (a,a_1)\cup(-a_1,-a)\cup O_1\cup O_2\setminus(\overline{\delta}^{a_1}\cup\overline{\delta}^b\cup\overline{\delta}^{-a_1}\cup\overline{\delta}^{-b}),\v\\
\mathbb{I}+\mathcal{O}(\dfrac{1}{t}),\quad z\in\partial\delta^{a_1}\cup\partial\delta^{b}\cup\partial\delta^{-a_1}\cup\partial\delta^{-b},
\end{cases}
\ene
with $c$ being a suitable positive real parameter.

 \item {} Normalization: $E(x,t;z)=\mathbb{I}+\mathcal{O}(z^{-1}),\quad z\rightarrow\infty.$
\end{itemize}
\end{prop}

According to Riemann-Hilbert problem \ref{RH14}, we obtain
\bee\label{E-error-a1}
E(x,t;z)=\mathbb{I}+\frac{E_1(x,t)}{tz}+\mathcal{O}(z^{-2}).
\ene

Based on Ref.~\cite{Girotti-1}, one has the following proposition:

\begin{propo}
As $t\to+\infty$, the potential function $q(x,t)$ of the complex mKdV equation has the following asymptotic behaviour:
\bee \label{asy-time}
q(x,t)=(b-a_1)\dfrac{\vartheta_3(0;\tau)\vartheta_3(\frac{tm_{2a_1}+m_{1a_1}}{2\pi i}+\frac12;\tau)}{\vartheta_3(\frac12;\tau)\vartheta_3(\frac{tm_{2a_1}+m_{1a_1}}{2\pi i};\tau)}+\mathcal{O}(t^{-1}),\quad t\to +\infty.
\ene

\end{propo}

\begin{proof}
According to Eqs.~(\ref{Y1-Y-a1}),(\ref{Y2-Y1-a1}) and (\ref{E-error-a1}), we have
\begin{align}\no
\begin{aligned}
Y_{a_1}(x,t;z)&=Y_{1a_1}(x,t;z)f_{a_1}(z)^{-\sigma_3}e^{-tg_{a_1}(z)\sigma_3}\v\\
&=Y_{2a_1}(x,t;z)f_{a_1}(z)^{-\sigma_3}e^{-tg_{a_1}(z)\sigma_3}\v\\
&=E(x,t;z)Y_{3a_1}(x,t;z)f_{a_1}(z)^{-\sigma_3}e^{-tg_{a_1}(z)\sigma_3}\v\\
&=\left(\mathbb{I}+\frac{E_1(x,t)}{tz}+\mathcal{O}(z^{-2})\right)Y_{3a_1}(x,t;z)f_{a_1}(z)^{-\sigma_3}e^{-tg_{a_1}(z)\sigma_3},
\end{aligned}
\end{align}
which further leads to
\begin{align}
\begin{aligned}
Y_{a_1,12}(x,t;z)&=\left(Y_{3a_1,12}(x,t;z)+\frac{E_{1,12}(x,t)}{tz}+\mathcal{O}(z^{-2})\right)f_{a_1}(z)e^{tg_{a_1}(z)}\v\\
&=\left(Y^{o}_{12}(x,t;z)+\frac{E_{1,12}(x,t)}{tz}+\mathcal{O}(z^{-2})\right)f_{a_1}(z)e^{tg_{a_1}(z)}\v\\
&=\frac12\left(\psi_{12}(x,t;z)+\frac{E_{1,12}(x,t)}{tz}+\mathcal{O}(z^{-2})\right)f_{a_1}(z)e^{tg_{a_1}(z)}.
\end{aligned}
\end{align}

Note that
\begin{align}
\begin{aligned}
&f_{a_1}(z)=1+\dfrac{f_1}{z}+\mathcal{O}(z^{-2}),\v\\
&e^{tg_{a_1}(z)}\to1,\quad z\to\infty,\v\\
&m_{5a_1}(z)=1+\dfrac{a_1-b}{2z}+\mathcal{O}(z^{-2}).
\end{aligned}
\end{align}

According to Eq.~(\ref{fanyan-6}), namely, $q(x,t)=-2\lim\limits_{z\rightarrow\infty}z Y_{a_1,12}(x,t;z),$ one can find that this Theorem holds.

\end{proof}

\subsection{The case $\xi=x/t<\xi_c$}

When $\xi<\xi_c$, we need to redefine the function $g(z)$, which satisfies the following properties~\cite{Girotti-1}:
\begin{itemize}

 \item {} Analyticity: $g_{a}(z)$ is analytic in $\mathbb{C}\setminus(-b,b)$ and takes continuous boundary values on $(-b,b)$.

 \item {} Jump condition: The boundary values on the jump contour are defined as
\begin{align}
\begin{aligned}
&g_{a+}(z)+g_{a-}(z)=2\xi z-8z^3,\quad z\in \gamma_{1}\cup\gamma_2,\v\\
&g_{a+}(z)-g_{a-}(z)=m_{2a},\quad z\in (-a,a),
\end{aligned}
\end{align}
where
\bee\no
m_{2a}=\dfrac{\pi ib(\xi-2\left(a^2+b^2\right))}{K_1(m)}.
\ene

 \item {} Normalization: $g_{a}(z)=\mathcal{O}(z^{-1}),\quad z\rightarrow\infty.$
\end{itemize}

Then, we have
\begin{align}
\begin{aligned}
&\mathrm{Re}(g_{a}(z)+4z^3-\xi z)>0,\quad z\in O_1\setminus\{a,b\},\v\\
&\mathrm{Re}(g_{a}(z)+4z^3-\xi z)<0,\quad z\in O_2\setminus\{-a,-b\}.
\end{aligned}
\end{align}

Based on Ref.~\cite{Girotti-1}, one has
\begin{propo}
As $t\to+\infty,\xi<\xi_c$, we have the long-time asymptotic behaviour of the soliton gas of the potential function $q(x,t)$ of the complex mKdV equation:
\bee
q(x,t)=(b-a)\dfrac{\vartheta_3(0;\tau)\vartheta_3(\frac{tm_{2a}+m_{1}}{2\pi i}+\frac12;\tau)}{\vartheta_3(\frac12;\tau)\vartheta_3(\frac{tm_{2a}+m_{1}}{2\pi i};\tau)}+\mathcal{O}(t^{-1}),\quad t\to +\infty.
\ene
\end{propo}

\section{Soliton gas in the quadrature domains}

Bertola {\it et al} studied the soliton shielding in the quadrature domains of the focusing nonlinear Schr\"odinger equation~\cite{Grava-3}.   Similarly, we consider the case that the discrete spectra $z_j\,(j=1,\cdots,N)$ with the norming constants $c_j\,(j=1,\cdots,N)$ filling uniformly compact domain $\Omega_1$ of the complex upper half space $\mathbb{C}_+$, for example~\cite{Grava-3},
\bee
\Omega_1:=\{z|~|(z-s_1)^{\ell}-s_2|<s_3\},\quad \Omega_1\subset\mathbb{C}_+,
\ene
where $\ell\in\mathbb{N}_+,s_1\in\mathbb{C}_+$ and $|s_2|,s_3$ are sufficiently small real constants.
The boundary of $\Omega_1^*$ which is complex conjugate domain of $\Omega_1$ is described by (see Fig.~\ref{fig7})
\bee
z^*=\left(s_2^*+\dfrac{s_3^2}{(z-s_1)^m-s_2}\right)^{1/\ell}+s_1^*,\quad z\in\Omega_1.
\ene

\begin{figure}[!t]
    \centering
 \vspace{-0.15in}
  {\scalebox{0.55}[0.55]{\includegraphics{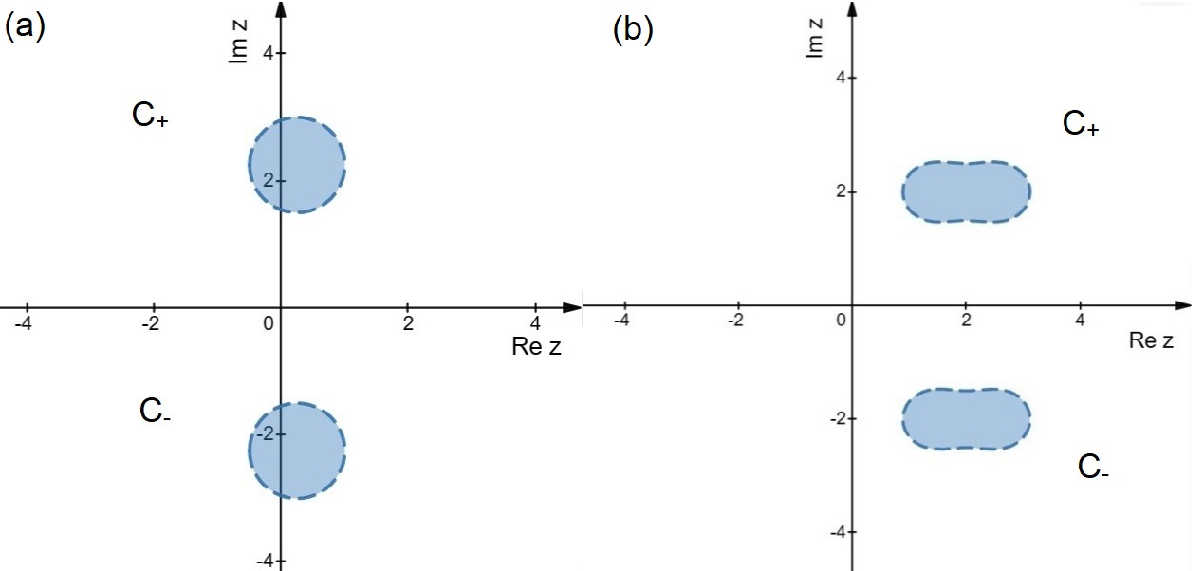}}}\hspace{-0.35in}
\vspace{0.1in}
\caption{The domains $\Omega_1$ and $\Omega_1^*$: (a) The parameters are $s_1=2i,\, s_2=\frac14+\frac{1}{4}i,s_3=\frac34,\, \ell=1$,~(b) The parameters are $s_1=2+2i,\,
s_2=\frac12,\, s_3=\frac34,\, \ell=2$.}
   \label{fig7}
\end{figure}

We assume that the norming constants $c_j\,(j=1,\cdots,N)$ have the following forms:
\bee
c_j=\frac{|\Omega_1|r(z_j,z_j^*)}{N\pi}.
\ene
where $|\Omega_1|$ means the area of the domain $\Omega_1$ and $r(z,z^*)$ is a smooth function with respect to variables $z$ and $z^*$. Then we construct the following Riemann-Hilbert problem.

\begin{prop}\label{RH1-1}
Find a $2\times 2$ matrix function $M_0(x,t;z)$ that satisfies the following properties:

\begin{itemize}

 \item {} Analyticity: $M_0(x,t;z)$ is meromorphic in $\mathbb{C}$, with simple poles at $K_1:=\{z_j,z_j^*\}_{j=1}^N$ with $z_j\in\Omega_1$.

 \item {} Normalization: $M_0(x,t;z)=\mathbb{I}+\mathcal{O}(z^{-1}),\quad z\rightarrow\infty$;

\item {} Residue conditions: $M_0(x,t;z)$ has simple poles at each point in $K_1$ with:
\begin{align}\label{Res12}
\begin{aligned}
&\mathrm{Res}_{z=z_j}M_0(x,t;z)=\lim\limits_{z\to z_j}M_0(x,t;z)\left[\!\!\begin{array}{cc}
0& 0  \vspace{0.05in}\\
c_je^{2i(zx+(2\alpha z^2+4\beta z^3)t)}& 0
\end{array}\!\!\right],\vspace{0.05in}\\
&\mathrm{Res}_{z=z_j^*}M_0(x,t;z)=\lim\limits_{z\to z_j^*}M_0(x,t;z)\left[\!\!\begin{array}{cc}
0& -c_j^*e^{-2i(zx+(2\alpha z^2+4\beta z^3)t)}  \vspace{0.05in}\\
0& 0
\end{array}\!\!\right].
\end{aligned}
\end{align}

\end{itemize}
\end{prop}

According to Eq.~(\ref{fanyan}), we can recover $q(x,t)$ by the following formula:
\bee\label{fanyan-1-1}
q(x,t)=2i\lim\limits_{z\rightarrow\infty}(z M_{0}(x,t;z))_{12}.
\ene

We define a closed curve $\Gamma_{1+}$ with a very small radius encircling the poles $\{z_j\}_{j=1}^N$ counterclockwise in the upper half plane $\mathbb{C}_+$, and a closed curve $\Gamma_{1-}$ with a very small radius encircling the poles $\{z_j^*\}_{j=1}^N$ counterclockwise in the lower half plane $\mathbb{C}_-$. Make the following transformation:
\bee
M_1(x,t;z)
=\begin{cases}
M_0(x,t;z)\left[\!\!\begin{array}{cc}
1& 0  \vspace{0.05in}\\
-\sum\limits_{j=1}^{N}\dfrac{c_je^{2i(xz+2\alpha tz^2+4\beta tz^3)}}{z-z_j}& 1
\end{array}\!\!\right],\quad z~\mathrm{within}~\Gamma_+,\vspace{0.05in}\\
M_0(x,t;z)\left[\!\!\begin{array}{cc}
1& \sum\limits_{j=1}^{N}c_j^*\dfrac{e^{-2i(xz+2\alpha tz^2+4\beta tz^3)}}{z-z_j^*}  \vspace{0.05in}\\
0& 1
\end{array}\!\!\right],\quad z~\mathrm{within}~\Gamma_-,\v\\
M_0(x,t;z),\quad \mathrm{otherwise}.
\end{cases}
\ene
Then matrix function $M_1(x,t;z)$ satisfies the following Riemann-Hilbert problem.

\begin{prop}\label{RH2-2}
Find a $2\times 2$ matrix function $M_1(x,t;z)$ that satisfies the following properties:

\begin{itemize}

 \item {} Analyticity: $M_1(x,t;z)$ is analytic in $\mathbb{C}\setminus(\Gamma_{1+}\cup\Gamma_{1-})$ and takes continuous boundary values on $\Gamma_{1+}\cup\Gamma_{1-}$.

 \item {} Jump condition: The boundary values on the jump contour $\Gamma_{1+}\cup\Gamma_{1-}$ are defined as
 \bee
M_{1+}(x,t;z)=M_{1-}(x,t;z)J_1(x,t;z),\quad z\in\Gamma_{1+}\cup\Gamma_{1-},
\ene
where
\bee
J_1(x,t;z)
=\begin{cases}
\left[\!\!\begin{array}{cc}
1& 0  \vspace{0.05in}\\
-\sum\limits_{j=1}^{N}\dfrac{c_je^{2i(xz_j+2\alpha tz_j^2+4\beta tz_j^3)}}{z-z_j}& 1
\end{array}\!\!\right],\quad z\in\Gamma_{1+},\vspace{0.05in}\\
\left[\!\!\begin{array}{cc}
1& \sum\limits_{j=1}^{N}c_j^*\dfrac{e^{-2i(xz_j^*+2\alpha tz_j^{*2}+4\beta tz_j^{*3})}}{z-z_j^*}  \vspace{0.05in}\\
0& 1
\end{array}\!\!\right],\quad z\in\Gamma_{1-};
\end{cases}
\ene

 \item {} Normalization: $M_1(x,t;z)=\mathbb{I}+\mathcal{O}(z^{-1}),\quad z\rightarrow\infty.$
\end{itemize}
\end{prop}

According to Eq.~(\ref{fanyan-1-1}), we can recover $q(x,t)$ by the following formula:
\bee\label{fanyan-2-1}
q(x,t)=2i\lim\limits_{z\rightarrow\infty}(z M_{1}(x,t;z))_{12}.
\ene

Based on Ref.\cite{Grava-3}, one has the following lemma:

\begin{lemma}\label{le1-1}
For any open set $B_+$ containing the domain $\Omega_1$, and any open set $B_-$ containing the domain $\Omega_1^*$, the following identities hold:
\bee
\lim\limits_{N\to\infty}\sum\limits_{j=1}^{N}\frac{c_je^{2i(xz_j+2\alpha tz_j^2+4\beta tz_j^3)}}{z-z_j}=\iint_{\Omega_1}\frac{r(\lambda,\lambda^*)}{2\pi i(z-\lambda)}
e^{2i(\lambda x+2\alpha \lambda^2t+4\beta \lambda^3t)}d\lambda^*\wedge d\lambda,
\ene
uniformly for all $\mathbb{C}\setminus B_+$,
\bee\label{integ-1}
\lim\limits_{N\to\infty}\sum\limits_{j=1}^{N}\frac{c_j^*e^{-2i(xz_j+2\alpha tz_j^2+4\beta tz_j^3)}}{z-z_j^*}=\iint_{\Omega_1^*}
\frac{r^*(\lambda,\lambda^*)}{2\pi i(z-\lambda)}e^{-2i(\lambda x+2\alpha \lambda^2t+4\beta \lambda^3t)}d\lambda^*\wedge d\lambda,
\ene
uniformly for all $\mathbb{C}\setminus B_-$. The boundary $\partial\Omega_1$ and $\partial\Omega_1^*$ are both counterclockwise.

\end{lemma}

According to Lemma \ref{le1-1}, we can obtain the Riemann-Hilbert problem of matrix function $M_1(x,t;z)$ as $N\to\infty$.

\begin{prop}\label{RH3-2g}
Find a $2\times 2$ matrix function $M_2(x,t;z)$ that satisfies the following properties:

\begin{itemize}

 \item {} Analyticity: $M_2(x,t;z)$ is analytic in $\mathbb{C}\setminus(\Gamma_{1+}\cup\Gamma_{1-})$ and takes continuous boundary values on $\Gamma_{1+}\cup\Gamma_{1-}$.

 \item {} Jump condition: The boundary values on the jump contour $\Gamma_{1+}\cup\Gamma_{1-}$ are defined as
 \bee
M_{2+}(x,t;z)=M_{2-}(x,t;z)J_2(x,t;z),\quad z\in\Gamma_{1+}\cup\Gamma_{1-},
\ene
where
\bee\label{J-2}
J_2(x,t;z)
=\begin{cases}
\left[\!\!\begin{array}{cc}
1& 0  \vspace{0.05in}\\
\d\iint_{\Omega_1}\frac{r(\lambda,\lambda^*)e^{2i(x\lambda+2\alpha t\lambda^2+4\beta t\lambda^3)}}{2\pi i(\lambda-z)}d\lambda^*\wedge d\lambda& 1
\end{array}\!\!\right],\quad z\in\Gamma_{1+},\vspace{0.05in}\\
\left[\!\!\begin{array}{cc}
1& \d\iint_{\Omega_1^*}\frac{r^*(\lambda,\lambda^*)e^{-2i(x\lambda+2\alpha t\lambda^{2}+4\beta t\lambda^{3})}}{2\pi i(z-\lambda)}d\lambda^*\wedge d\lambda
  \vspace{0.05in}\\
0& 1
\end{array}\!\!\right],\quad z\in\Gamma_{1-};
\end{cases}
\ene

 \item {} Normalization: $M_1(x,t;z)=\mathbb{I}+\mathcal{O}(z^{-1}),\quad z\rightarrow\infty.$

\end{itemize}
\end{prop}

Let the function $R(z,z^*)$ satisfy the following property:
\bee
\overline{\partial}R(z,z^*)=r(z,z^*).
\ene
Then the jump matrix $J_2(x,t;z)$ can be written as:
\bee\label{J-2-1}
J_2(x,t;z)
=\begin{cases}
\left[\!\!\begin{array}{cc}
1& 0  \vspace{0.05in}\\
\d\int_{\partial\Omega_1}\frac{R(\lambda,\lambda^*)e^{2i(x\lambda+2\alpha t\lambda^2+4\beta t\lambda^3)}}{2\pi i(\lambda-z)}d\lambda& 1
\end{array}\!\!\right],\quad z\in\Gamma_{1+},\vspace{0.05in}\\
\left[\!\!\begin{array}{cc}
1& \d\int_{\partial\Omega_1^*}\frac{R^*(\lambda,\lambda^*)e^{-2i(x\lambda+2\alpha t\lambda^{2}+4\beta t\lambda^{3})}}{2\pi i(z-\lambda)}d\lambda
  \vspace{0.05in}\\
0& 1
\end{array}\!\!\right],\quad z\in\Gamma_{1-};
\end{cases}
\ene

Based on Ref.\cite{Grava-3}, one has the following proposition:

\begin{propo} If we assume that the function $R(z, z^*)=R(z,s_1^*+\left(s_2^*+\dfrac{s_3^2}{(z-s_1)^{\ell}-s_2}\right)^{\frac{1}{\ell}})$ in $\Omega_1$ has $n$ distinct poles $\lambda_j,j=1,\cdots,n$, that is,
\bee\label{assume-1}
R\left(z,s_1^*+\left(s_2^*+\dfrac{s_3^2}{(z-s_1)^{\ell}-s_2}\right)^{\frac{1}{\ell}}\right)=\widehat R(z)\prod_{j=1}^n\frac{1}{z-\lambda_j},
\ene
where $\widehat R(z)$ is some analytic function in $\Omega_1$,
then the solution of the Riemann-Hilbert problem~\ref{RH3-2g} can generate an $n$-soliton $q_n(x,t)$ with discrete spectrums $\lambda_j,j=1,\cdots,n$ with the norming constants $c_j=\dfrac{R(\lambda_j)}{\prod_{k\neq j}(\lambda_j-\lambda_k)},j=1,\cdots,n$.
\end{propo}

\begin{proof}

Substituting Eq.~(\ref{assume-1}) into the integral in Eq.~(\ref{J-2-1}), we obtain
\begin{align}\no
\begin{aligned}
&\int_{\partial\Omega_1}\frac{R(\lambda,\lambda^*)e^{2i(x\lambda+2\alpha t\lambda^2+4\beta t\lambda^3)}}{2\pi i(\lambda-z)}d\lambda\v\\
=&\int_{\partial\Omega_1}\frac{R\left(\lambda,s_1^*+\left(s_2^*+\dfrac{s_3^2}{(\lambda-s_1)^{\ell}-s_2}\right)^{\frac{1}{\ell}}\right)
}{2\pi i(\lambda-z)}e^{2i(x\lambda+2\alpha t\lambda^2+4\beta t\lambda^3)} d\lambda\v\\
=&\int_{\partial\Omega_1}\frac{\widehat R(\lambda)e^{2i(x\lambda+2\alpha t\lambda^2+4\beta t\lambda^3)}}{2\pi i(z-\lambda)\prod_{j=1}^n(\lambda-\lambda_j)}d\lambda\v\\
=&\sum_{j=1}^n\dfrac{\widehat R(\lambda_j)e^{2i(x\lambda_j+2\alpha t\lambda_j^2+4\beta t\lambda_j^3)}}{(z-\lambda_j)\prod_{k\neq j}(\lambda_j-\lambda_k)},\quad z\notin\Omega_1.
\end{aligned}
\end{align}

Let the norming constants be $c_j=\dfrac{\widehat R(\lambda_j)}{\prod_{k\neq j}(\lambda_j-\lambda_k)},j=1,\cdots,n$, then, similarly to the Riemann-Hilbert problem \ref{RH1}, the solution of Riemann-Hilbert problem \ref{RH3-2g} leads to an $n$-soliton $q_n(x,t)$ of the Hirota equation. Thus the proof is completed.
\end{proof}


\section{Conclusions and discussions}

In conclusion, we have used the idea of Refs.~\cite{Girotti-1,Girotti-2,Grava-3} to investigate the asymptotic behaviors of a soliton gas (the limit $N\to \infty$ of $N$-soliton solutions) of the Hirota equation including the complex modified KdV equation with the aid of the Riemann-Hilbert problems with pure imaginary discrete spectra restricted in an interval. We show that that this soliton gas tends slowly to the Jaocbian elliptic wave solution (zero exponentially quickly ) when $x\to -\infty$ ($x\to +\infty$). Moreover, we also give the long-time asymptotics of the soliton gas under the different velocity conditions: $x/t>4\beta b^2,\, \xi_c<x/t<4\beta b^2,\, x/t<\xi_c$. Finally, we analyze the property of the soliton gas with the discrete spectra filling uniformly a quadrature domain.

Therefore are some challenging issues, for example, i) when $\xi=x/t<4\beta b^2$, how will one consider the long-time asymptotics of the potential $q(x,t)$ of the Hirota equation for the case $\alpha\beta\not=0$ ?
since there is the terms $iR(z)e^{\pm 2zt(\xi-4\beta z^2+2i\alpha z)}$ with the factor $e^{\pm 4i\alpha z^2t}$ in the jump matrices given by Eq.~(\ref{M3-3}); ii) how will one study the asymptotics of soliton gas for the case of the considered discrete spectral intervals localized in the non-imaginary axes ?

\vspace{0.2in}
\noindent {\bf Acknowledgments}

\addcontentsline{toc}{section}{Acknowledgments}

\vspace{0.05in}
This work was supported by the National Natural Science Foundation of China (Grant No. 11925108).


\begin{thebibliography}{100}

\addcontentsline{toc}{section}{References}
{\small

\bibitem{soliton} N. J. Zabusky, M. D. Kruskal, Interaction of ``solitons" in a collisionless plasma and the recurrence of initial states, Phys. Rev. Lett. 15, 240 (1965).

\bibitem{GGKM} C. S. Gardner, J.M. Greene, M.D. Kruskal, and R.M. Miura, Method for solving the Korteweg-de Vries equation, Phys. Rev Lett. 19, 1095-1097 (1967).

\bibitem{AKNS} M.J. Ablowitz, D.J. Kaup, A.C. Newell, and H. Segur, The inverse scattering transform. Fourier analysis for nonlinear problems, Stud. Appl. Math. 53, 249-315 (1974).

\bibitem{Lax} P.D. Lax, Integrals of nonlinear evolution equations and solitary waves, Comm. Pure Appl. Math. 21, 467-490 (1968).

\bibitem{Nov84} S.P. Novikov, S.V. Manakov, L.P. Pitaevskii, and V.E.Zakharov, {\it Theory of Solitons: The Inverse ScatteringTransform} (Plenum, New York, 1984).

\bibitem{ACN} A.C. Newell, {\it Solitons in Mathematics and Physics} (SIAM, Philadelphia, 1985).

\bibitem{FT} L. D. Faddeev, L. A. Takhtajan, {\it Hamiltonian Methods in the Theory of Solitons} (Springer-Verlag, Berlin, 1987).

\bibitem{AC91} M.J. Ablowitz and P.A. Clarkson, {\it Solitons, Nonlinear Evolution Equations and Inverse Scattering} (Cambridge University Press, Cambridge, 1991).

\bibitem{ACS} A.C. Scott, {\it Nonlinear Science: Emergence and Dynamics of Coherent Structures} (Oxford University Press, Oxford, 2003).

\bibitem{book1} G. P. Agrawal, {\em Nonlinear Fiber Optics} (5th edn.) (Academic Press, New York, 2012).

\bibitem{exp} L.F. Mollenauer, R.H. Stolen, J. P. Gordon, Experimental observation of picosecond pulse narrowing and solitons in
optical fibers, Phys. Rev. Lett. 45, 1095-1112 (1980).

\bibitem{book2} Y. S. Kivshar and  G. P. Agrawal, {\em  Optical Solitons: from Fibers to Photonic Crystals} (Academic Press, New York, 2013).

\bibitem{nls05} B. A. Malomed, D. Mihalache, F. Wise, and L. Torner, Spatiotemporal optical solitons, J. Opt. B: Quantum Semiclass. Opt. 7, R53 (2005).

\bibitem{book3} L. Pitaevskii and S. Stringari, {\em Bose-Einstein Condensation and Superfluidity} (Oxford University Press, Oxford, 2016).

\bibitem{zak71}  V.E. Zakharov, Kinetic equation for solitons. Sov. Phys. JETP 33, 538-541 (1971).

\bibitem{E1} G.A. El, A.M. Kamchantov, Kinetic equation for a dense soliton gas, Phys. Rev. Lett. 95, 204101
(2005).

\bibitem{E2} G.A. El, A.M. Kamchatnov, M.V. Pavlov, S.A. Zykov,  Kinetic equation for a soliton gas and
its hydrodynamic reductions, J. Nonlinear Sci. 21, 151-191 (2011).

\bibitem{E3} G.A. El, Critical density of a soliton gas, Chaos 26, 023105 (2016).

\bibitem{E4} G.A. El, A. Tovbis, Spectral theory of soliton and breather gases for the focusing nonlinear
Schr\"odinger equation, Phys. Rev. E 101, 052207 (2020).

\bibitem{E7} P. Suret, A. Tikan, F. Bonnefoy, F. Copie, G. Ducrozet, A. Gelash, G. Prabhudesai, G. Michel, A. Cazaubiel, E. Falcon,
G. A. El, S. Randoux, Nonlinear spectral synthesis of soliton gas in deep-water surface gravity waves, Phys. Rev. Lett. 125, 264101
(2020).

\bibitem{E5} T. Congy, G. El, G. Roberti, Soliton gas in bidirectional dispersive hydrodynamics, Phys. Rev. E 103,
042201 (2021).

\bibitem{E6} E. G. Shurgalina, E. N. Pelinovsky, Nonlinear dynamics of a soliton gas: Modified Korteweg-de Vries equation
framework, Phys. Lett. A 380, 2049-2053 (2016).

\bibitem{Girotti-1} M. Girotti, T. Grava, R. Jenkins, K. T. R. McLaughlin, Rigorous asymptotics of a KdV soliton gas, Commun. Math. Phys. 384, 733-784 (2021).

\bibitem{Girotti-2} M. Girotti, T. Grava, R. Jenkins, K. T. R. McLaughlin, A. Minakov, Soliton versus the gas: Fredholm determinants, analysis, and the rapid oscillations behind the kinetic equation, Comm. Pure Appl. Math. 76, 3233-3299 (2023).

\bibitem{Grava-3} M. Bertola, T. Grava, G. Orsatti, Soliton shielding of the focusing Nonlinear Schr\"odinger Equation, Phys. Rev. Lett. 130, 127201 (2023).
    
\bibitem{miller} G. D. Lyng, P. D. Miller,  The $N$-soliton of the focusing nonlinear Schr\"odinger equation for $N$ large,
 Comm. Pure Appl. Math. 60, 951-1026 (2007).

\bibitem{bil1} D. Bilman, R. Buckingham,  Large-order asymptotics for multiple-pole solitons of the
focusing nonlinear Schr\"odinger equation. J. Nonlinear Sci. 29, 2185-2229 (2019).

\bibitem{bil2} D. Bilman, L. Ling, P. D. Miller, Extreme superposition: rogue waves of infinite order and the Painlev\'e-III hierarchy, Duke Math. J. 169, 671-760 (2020).

\bibitem{bil3} D. Bilman, R. Buckingham, D. S. Wang, Far-field asymptotics for multiple-pole solitons in the large-order limit, J. Differ. Equ. 297, 320 (2021).


\bibitem{RIST} D. Bilman, P. D. Miller, A robust inverse scattering transform for the focusing nonlinear Schr\"odinger equation, Comm. Pure Appl. Math. 72, 1722 (2019).

\bibitem{RH2} X. Zhou, The Riemann-Hilbert problem and inverse scattering, SIAM J. Math. Anal. 20, 966 (1989).

\bibitem{RH} P. Deift, X. Zhou, A steepest descent method for oscillatory Riemann-Hilbert problems, Asymptotics for the MKdV equation, Ann. of Math. 137, 295 (1993).

\bibitem{dbar} K.T.R. McLaughlin, P.D. Miller, The $\bar\partial$-steepest descent method and the asymptotic behavior of polynomials orthogonal on the unit circle with fixed and exponentially varying non-analytic weights, Int. Math. Res. Not. 2006,
48673 (2006).

\bibitem{dbar2} K.T.R. McLaughlin, P.D. Miller, The $\bar\partial$-steepest descent method for orthogonal polynomials on the real line with varying weights, Int. Math. Res. Not. 075 (2008).


\bibitem{RHP1} R. Beals, R. Coifman, Scattering and inverse scattering for first order systems, Comm. Pure Appl. Math. 37, 39-90 (1984).

\bibitem{RHP2} R. Beals, P. Deift, C. Tomei, {\it Direct and Inverse Scattering on the Line} (American Mathematical Society,
Providence, 1988).

\bibitem{RHP3} T. Trogdon, S. Olver, {\it Riemann-Hilbert Problems, Their Numerical Solution, and the Computation of Nonlinear Special Functions} (SIAM, Philadelphia,  2016).

\bibitem{n1} P. Deift, X. Zhou, Long-time asymptotics for solutions of the NLS equation with initial data in a
weighted Sobolev space, Comm. Pure Appl. Math. 56, 1029-1077 (2003).

\bibitem{n1a} A. H. Vartanian, Long-time asymptotics of solutions to the Cauchy problem for the defocusing nonlinear Schr\"odinger equation with finite-density initial data. II. Dark solitons on continua. Math. Phys. Anal. Geom. 5, 319-413 (2002).

\bibitem{n2} A. Boutet de Monvel, A. Its, V. Kotlyarov, Long-time asymptotics for the focusing NLS equation
with time-periodic boundary condition on the half-line, Commun. Math. Phys. 290, 479-522 (2009).

\bibitem{n3} K. Grunert, G. Teschl, Long-time asymptotics for the Korteweg de Vries equation via nonlinear
steepest descent, Math. Phys. Anal. Geom. 12, 287-324 (2009).

\bibitem{n4} A. Boutet de Monvel, J. Lenells, D. Shepelsky, Long-time asymptotics for the Degasperis-Procesi
equation on the half-line, Ann. Inst. Fourier 69, 171-230 (2019).

\bibitem{n5} J. Xu, E. Fan, Long-time asymptotics for the Fokas-Lenells equation with decaying initial value problem: without solitons, J. Differ. Equ. 259, 1098-1148 (2015).

\bibitem{n6} H. Liu, X. Geng, B. Xue, The Deift-Zhou steepest descent method to long-time asymptotics for
the Sasa-Satsuma equation, J. Differ. Equ. 265, 5984-6008 (2018).

\bibitem{n7} G. Biondini, S. Li, D. Mantzavinos, Long-time asymptotics for the focusing nonlinear Schr\"odinger equation with nonzero boundary conditions in the presence of a discrete spectrum, Commun. Math. Phys. 382, 1495-1577  (2021).

\bibitem{d1} S. Cuccagna, R. Jenkins, On the asymptotic stability of N-soliton solutions of the defocusing nonlinear Schr\"odinger equation, Commun. Math. Phys. 343, 921-969 (2016).

\bibitem{d2} M. Borghese, R. Jenkins, K.T.R. McLaughlin, Long-time asymptotic behavior of the focusing nonlinear Schr\"odinger equation, Ann. Inst. Henri Poincar\'e Anal. 35, 887-920 (2018).

\bibitem{d3} R. Jenkins, J. Liu, P. Perry, C. Sulem, Soliton resolution for the derivative nonlinear Schr\"odinger equation, Commun. Math. Phys. 363,1003-1049 (2018).

\bibitem{d4} Y. Yang, E. Fan, Soliton resolution for the short-pulse equation, J. Differ. Equ. 280, 644-689 (2021).

\bibitem{d5} Y. Yang, E. Fan, On the long-time asymptotics of the modified Camassa-Holm equation in space-time solitonic regions. Adv. Math. 402, 108340 (2022).

\bibitem{d6}Z. Wang, E. Fan, The defocusing nonlinear Schr\"odinger equation with a nonzero background:: Painlev\'e asymptotics in two transition regions, Commun. Math. Phys. 402, 2879-2930 (2023).

\bibitem{d7} Z. Li,  S. Tian, J. Yang, On the asymptotic stability of N-soliton solution for the short pulse equation with weighted Sobolev initial data, J. Differ. Equ. 377, 121-187 (2023).



\bibitem{hnls} Y. Kodama, Optical solitons in a monomode fiber, J. Stat. Phys. 39, 597 (1985).

\bibitem{hnls2} Y.  Kodama  and  A. Hasegawa, Nonlinear pulse propagation in a monomode dielectric guide, IEEE J. Quantum Electron.
23, 510 (1987).

\bibitem{yan13} Z. Yan and C. Dai, Optical rogue waves in the generalized inhomogeneous higher-order nonlinear Schr\"odinger equation with modulating
coefficients, J. Opt. 15, 064012 (2013).

\bibitem{SS}N. Sasa, J. Satsuma, New-type of soliton solutions for a high-order nonlinear Schr?dinger equation, J. Phys. Soc.
Jpn. 60, 409-417 (1990).


\bibitem{ZS} V.E. Zakharov, A.B. Shabat, Exact theory of two-dimensional self-focusing and one-dimensional self-modulation
of waves in nonlinear media, Sov. Phys. JETP 34, 62-69 (1972).

\bibitem{wadati79} M. Wadati, K. Konno, Y. H.Ichikawa, A generalization of inverse scattering method, J. Phys. Soc. Jpn.
46, 1965 (1979).

\bibitem{KN} D. J. Kaup, A. C. Newell, An exact solution for a derivative nonlinear Schr?dinger equation, J. Math. Phys.
19, 798-801 (1978).

\bibitem{CLL} H.H. Chen, Y. C. Lee, C.S. Liu, Integrability of nonlinear Hamiltonian systems by inverse scattering method, Phys. Scr. 20, 490 (1979).

\bibitem{Kundu} A. Kundu,  Landau-Lifshitz and higher-order nonlinear systems gauge generated from nonlinear Schr\"odinger-type equations, J. Math. Phys. 25, 3433 (1984).



\bibitem{hirota} R. Hirota, Exact envelope-soliton solutions of a nonlinear wave equation, J. Math. Phys. 14, 805 (1973).

\bibitem{h1} A. Ankiewicz, J. M. Soto-Crespo, and N. Akhmediev, Rogue waves and rational solutions of the Hirota equation, Phys. Rev. E 81, 046602 (2010).

\bibitem{h2}A. Chowdury, D. J. Kedziora, A. Ankiewicz, and N. Akhmediev, Soliton solutions of an integrable nonlinear Schr\"odinger equation with quintic terms, Phys. Rev. E 90, 032922 (2014).

\bibitem{h3} Y. Tao, J. He, Multisolitons, breathers, and rogue waves for the Hirota equation generated by the Darboux transformation, Phys. Rev. E 85, 026601 (2012).

\bibitem{h4} Y. Yang, Z. Yan, B. A. Malomed, Rogue waves, rational solitons, and modulational instability in an integrable fifth-order nonlinear Schr\"odinger equation, Chaos 25, 103112 (2015).

\bibitem{h5} S. Chen, Z. Yan, The Hirota equation: Darboux transform of the Riemann-Hilbert problem and higher-order rogue waves, Appl. Math. Lett. 95, 65-71 (2019).


\bibitem{zhang20}G. Zhang, S. Chen, Z. Yan, Focusing and defocusing Hirota equations with non-zero
boundary conditions: Inverse scattering transforms and soliton solutions, Commun. Nonlinear Sci. Numer. Simulat. 80 (2020) 104927.

\bibitem{hs2} L. Huang, J. Xu, E. Fan,  Long-time asymptotic for the Hirota equation via nonlinear steepest descentmethod, Nonlinear Anal. Real World Appl. 26, 229-262 (2015).

\bibitem{hs3} B. Guo,  N. Liu, Y. Wang,  Long-time asymptotics for the Hirota equation on the half-line, Nonlinear Anal. 174, 118-140 (2018).

\bibitem{hs1} S. Chen, Z. Yan, B. Guo, Long-time asymptotics for the focusing Hirota equation with non-zero boundary conditions at infinity via the Deift-Zhou approach, Math. Phys. Anal. Geom. 24, 17 (2021).



\bibitem{guo23}Y. Xiao, B. Guo, Z. Wang,  Nonlinear stability of multi-solitons for the Hirota equation, J. Differential Equa. 342, 369-417 (2023).

\bibitem{Kuijlaars-1} A. B. J. Kuijlaars, K. R. McLaughlin, W. Van Assche, M. Vanlessen, The Riemann-Hilbert approach to strong asymptotics for orthogonal polynomials on [-1,1], Adv. Math. 188, 337-398 (2004).

\bibitem{Deift-1} P. Deift, T. Kriecherbauer, K. T. R. McLaughlin, S. Venakides, X. Zhou, Strong asymptotics of orthogonal polynomials with respect to exponential weights, Comm. Pure Appl. Math. 52, 1491-1552 (1999).


\bibitem{Lawden-1} D. F. Lawden, {\it Elliptic Functions and Applications} (Springer, Berlin, 1989).

\bibitem{kdv-13} I. Egorova,  Z. Gladka, V. Kotlyarov, G. Teschl, Long-time asymptotics for the Korteweg-de
Vries equation with step-like initial data, Nonlinearity 26, 1839-1864 (2013).
     }

\end{thebibliography}
\end{document}